\documentclass[a4paper,english]{lipics-v2016}
\usepackage{cleveref}
\usepackage{microtype}
\usepackage{microtype}

\usepackage{todonotes}

\newcommand{\eucl}[1]{\left\lVert #1 \right\rVert}
\newcommand{\skipSpace}{\hfill \\ \noindent}
\newcommand{\set}[1]{\{ #1 \}}
\newcommand{\bbtl}[1]{bb_{t\ell}(#1)}
\newcommand{\bbtr}[1]{bb_{tr}(#1)}
\newcommand{\bbbl}[1]{bb_{b\ell}(#1)}
\newcommand{\bbbr}[1]{bb_{br}(#1)}
\newcommand{\twoDel}{2-localized Delaunay Graph}
\newcommand{\mixedChord}{3.56}
\newcommand{\onlineRoutingConstant}{7.87}
\newcommand{\algorithmConstant}{18.55}
\newcommand{\onlineRoutingConstantIntersecting}{12.83}

\newtheorem{lem}{\bf{Lemma}}{\bfseries}{\itshape}
\newtheorem{defi}{\bf{Definition}}{\bfseries}{\itshape}
\crefname{lem}{Lemma}{lemmas}
\Crefname{lem}{Lemma}{Lemmas}
\crefname{def}{definition}{definitions}
\Crefname{def}{Definition}{Definitions}

\title{A Bounding Box Overlay for Competitive Routing in Hybrid Communication Networks\footnote{This work was partially supported by the German Research Foundation (DFG) within the Collaborative Research Center ’On-The-Fly Computing’ (SFB 901).}}
\titlerunning{A Bounding Box Overlay for Hybrid Communication Networks} 

\author[1]{Jannik Castenow}
\author[2]{Christina Kolb}
\author[3]{Christian Scheideler}
\affil[1]{Heinz Nixdorf Institute \& Computer Science Department, Paderborn University, Paderborn, Germany\\
	\texttt{jannik.castenow@upb.de}}
\affil[2]{Computer Science Department, Paderborn University, Paderborn, Germany\\
  \texttt{ckolb@mail.upb.de}}
\affil[3]{Computer Science Department, Paderborn University, Paderborn, Germany\\
	\texttt{scheideler@upb.de}}

\authorrunning{C. Kolb, C. Scheideler \& J. Sundermeier} 


\EventEditors{John Q. Open and Joan R. Acces}
\EventNoEds{2}
\EventLongTitle{42nd Conference on Very Important Topics (CVIT 2016)}
\EventShortTitle{CVIT 2016}
\EventAcronym{CVIT}
\EventYear{2016}
\EventDate{December 24--27, 2016}
\EventLocation{Little Whinging, United Kingdom}
\EventLogo{}
\SeriesVolume{42}
\ArticleNo{}

\begin{document}

\maketitle

\begin{abstract}
In this work, we present a new approach for competitive geometric routing in wireless ad hoc networks. 
In general, it is well-known that any online routing strategy performs very poor in the worst case.
The main difficulty are uncovered regions within the wireless ad hoc network, which we denote as radio holes. 
Complex shapes of radio holes, for example zig-zag-shapes, make local geometric routing even more difficult, i.e., forwarded messages in direction to the destination might get stuck in a dead end or are routed along very long detours, when there is no knowledge about the ad hoc network. 
To obtain knowledge about the position and shape of radio holes, we make use of a hybrid network approach. 
This approach assumes that we can not just make use of the ad hoc network but also of some cellular infrastructure, which is used to gather knowledge about the underlying ad hoc network. 
Communication via the cellular infrastructure incurs costs as cell phone providers are involved.
Therefore, we use the cellular infrastructure only to compute routing paths in the ad hoc network.
The actual data transmission takes place in the ad hoc network.
In order to find good routing paths we aim at computing an abstraction of the ad hoc network in which radio holes are abstracted by bounding boxes.
The advantage of bounding boxes as hole abstraction is that we only have to consider a constant number of nodes per hole.
We prove that bounding boxes are a suitable hole abstraction that allows us to find $c$-competitive paths in the ad hoc network in case of non-intersecting bounding boxes.
In case of intersecting bounding boxes, we show via simulations that our routing strategy significantly outperforms the so far best online routing strategies for wireless ad hoc networks.
Finally, we also present a routing strategy that is $c$-competitive in case of pairwise intersecting bounding boxes.
\end{abstract}

\section{Introduction}
Imagine yourself walking through the city center with your smartphone. Because there are crowds of people with smartphones walking around as well, the density of smartphones is very high. In practice, whenever there are smartphones close by, i.e., one smartphone is in the WiFi range of another phone and vice versa, they can be connected via freely available direct wireless connections (e.g., WiFi Direct or Bluetooth). 
Thus, one can set up a wireless ad hoc network between smartphones, where the direct wireless communication mode enables the phones to send large amounts of data to each other.
We assume routing in the ad hoc network to be for free as messages are transmitted directly and no third party is involved.

In general, it would be much easier to communicate only via a cellular network since every node would be able to directly communicate with every other node (given that the cell phone infrastructure covers all nodes).
This, however, is only possible up to a limited amount of data.
Usually, smartphone owners have a contract with cellphone providers which offers a limited data volume.
Once the data volume has been exceeded, messages can only be exchanged at very low speed in the cellular network.
To maximize the lifetime of the data contracts while also being able to exchange almost unlimited data, it is evident to exchange all data via the ad hoc network whereas the cellular infrastructure is only used to find nearly optimal routing paths.
Finding nearly optimal routing paths in the ad hoc network is a non-trivial task, since sparse regions of the ad hoc network can lead to radio holes.
In general, natural and human made obstacles like buildings cause radio holes in the ad hoc network of smartphones. 
Complex shapes of radio holes, e.g., zig-zag shapes, make competitive local roting extremely difficult~\cite{DBLP:conf/mobihoc/KuhnWZ03}. 
Messages that are simply forwarded into the direction of the destination might get stuck in a dead end or are routed on very long detours, when there is no knowledge about the ad hoc network.
Unfortunately, collecting global knowledge about the entire ad hoc network, i.e., knowledge about the exact location and shape of radio holes, would be too expensive when only using cellular communication since potentially many people are located on the boundaries of holes. 
Therefore, we address the following question: Can cellular communication be used effectively to find near-shortest paths in the ad hoc network? \\

The authors in \cite{algosensorsPaper} were the first to provide an approach that combines the ad hoc and the cellular communication mode in a Hybrid Communication Network. 
The cellular infrastructure is used to establish an overlay network that computes the convex hulls of each radio hole to find near shortest routing paths that only consist of ad hoc links (i.e., the WiFi-interfaces of nodes).
They assume that convex hulls of holes do not intersect.
In this work, we improve their results concerning two aspects.
We propose an overlay network based on bounding boxes which only requires a constant number of nodes per hole for finding routing paths. 
Moreover, we consider intersections of bounding boxes, where the difficulty is that a path leading through an area of intersecting bounding boxes can be arbitrarily complex (see \Cref{fig:halloween2}).
We prove both theoretically and with simulations that our approach outperforms classical online routing strategies for geometric ad hoc routing significantly.

\begin{figure}[h]
	\centering
	\includegraphics[width=0.9\textwidth]{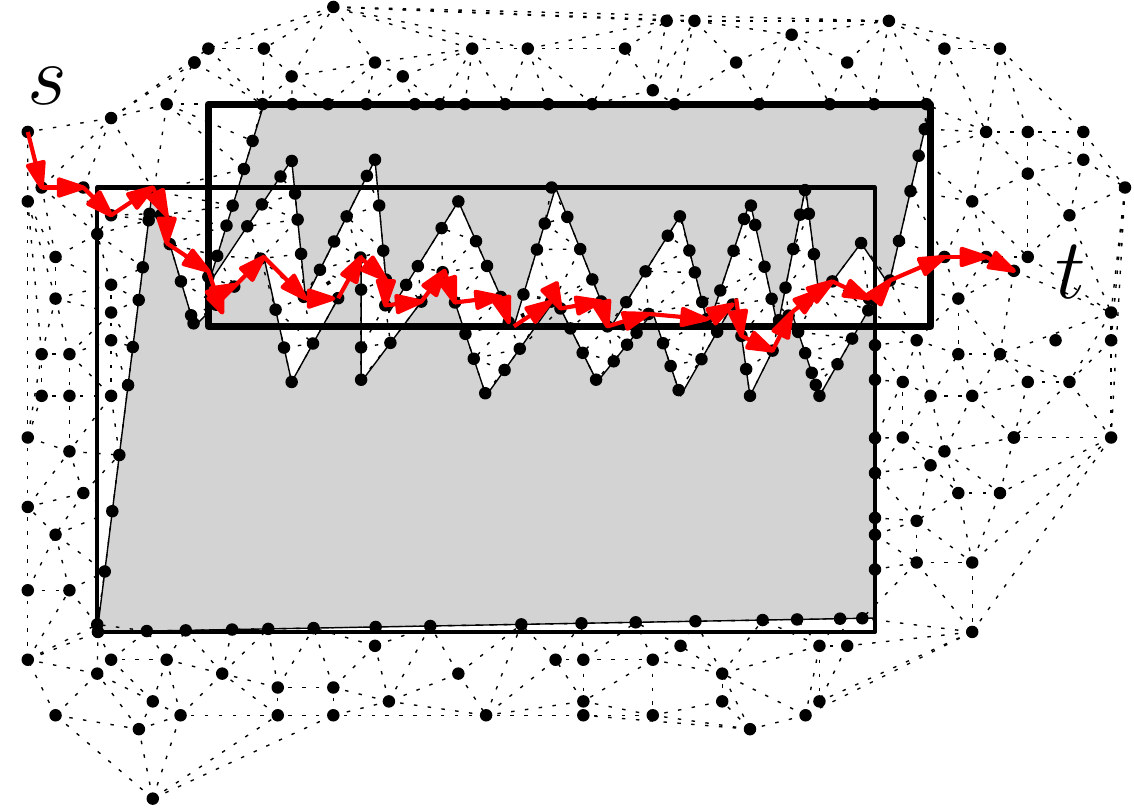}
	\caption{The shortest path between $s$ and $t$ (red arrows) leads through an area of intersecting bounding boxes. This path can be arbitrarily complex since the shapes of the radio holes can interleave each other as depicted in the figure.}
	\label{fig:halloween2} 
\end{figure}

\subsection{Motivation}
Classical research about ad hoc networks always assumes that no global infrastructure exists and nodes are only able to communicate via their ad hoc interface.
In this work, we consider a real-life scenario of people using their smartphones in urban areas.
A smartphone can be interpreted as ad hoc device since it is able to communicate for instance via WiFi Direct or Bluetooth.
Smartphones, however, are different as they combine multiple communication modes in one device and hence are also able to use a cellular infrastructure to communicate.
In metropolitan areas, the density of smartphones is sufficiently high such that the ad hoc network of smartphones is connected and, in principle, the entire data transmission between smartphones could be solely carried out via ad hoc links.
The main challenge of data transmission in an ad hoc network of smartphones in a metropolitan area is routing.
Usually, human made obstacles like buildings interfere wireless communication leading to radio holes in the ad hoc network.
It is well known that routing in an ad hoc network without knowledge about shapes and locations of radio holes potentially leads to very long detours \cite{DBLP:conf/podc/KuhnWZZ03}.
To overcome this drawback, the locations and shapes of radio holes can be determined efficiently via the cellular infrastructure which is available in metropolitan areas.
For a fast computation of routing paths, it makes sense to not consider the exact shape of a hole but a more coarse-grained abstraction like a bounding box.
In metropolitan areas, buildings are usually convex and rectangular shaped such that the bounding boxes of radio holes only rarely intersect, rendering the bounding box a practical hole abstraction in the described scenario.

\subsection{Model}
The model and definitions of this work are close to those of \cite{algosensorsPaper}.
We model the participants of the network as a set of nodes $V \subset \mathbb{R}^2$ in the Euclidean plane, where $|V| =n$.
Each node is associated with a unique ID (e.g., its phone number).
For any given pair of nodes $u,v$, we denote the Euclidean distance between $u$ and $v$ by $\|uv\|$.
Nodes have different communication modes.
For short distances, they can communicate via their WiFi-interface. 
Additionally, a node can communicate with every other node whose ID is known via the cellular infrastructure.
More formally, we model our network as a hybrid directed graph $H=\left(V,E,E_{AH}\right)$ where the node set $V$ represents the set of cell phones, an edge $\left(v,w\right)$ is in $E$ whenever $v$ knows the phone number (or simply {\em ID}) of $w$, and an edge $\left(v,w\right) \in E$ is also in the {\em ad hoc edge set} $E_{AH}$ whenever $v$ can send a message to $w$ using its Wifi interface.
For all edges $\left(v,w\right) \in E \setminus E_{AH}$, $v$ can only use a long-range link to directly send a message to $w$.
The concrete edges contained in $E$ and $E_{AH}$ are defined in later sections.
Since WiFi-communication can only be used over short distances, $E_{AH}$ can only contain edges which are part of the Unit Disk Graph of $V$ ($\mathrm{UDG}(V)$).
$\mathrm{UDG}\left(V\right)$, is a bi-directed graph that contains all edges $\left(u,v\right)$ with $\vert\vert uv\vert \vert\leq 1$. Assume $\mathrm{UDG}\left(V\right)$ to be connected so that a message can be sent from every node to every other node in $V$ by just using ad hoc edges. 

While the potential ad hoc edges are fixed, the nodes can change $E$ over time:
If a node $v$ knows the IDs of nodes $w$ and $w'$, then it can send the ID of $w$ to $w'$, which adds $\left(w,w'\right)$ to $E$.
This procedure is called \emph{ID-introduction}.
Alternatively, if $v$ deletes the address of some node $w$ with $\left(v,w\right) \in E$, then $\left(v,w\right)$ is removed from $E$.
There are no other means of changing $E$, i.e., a node $v$ cannot learn about an ID of a node $w$ unless $w$ is in $v$'s UDG-neighborhood or the ID of $w$ is sent to $v$ by some other node. \\
Moreover, we consider synchronous message passing in which time is divided into rounds.
We assume that every message initiated in round $i$ is delivered at the beginning of round $i+1$.

\subsection{Objective}
Our objective is to design a competitive routing algorithm for ad hoc networks, where the source $s$ of a message knows the ID of the destination $t$, or in other words, $\left(s,t\right) \in E$.
We call a routing strategy $c$-competitive, if the length of a path obtained by the strategy has length at most $c$ times the length of an optimal path for a constant $c$.
The authors in \cite{DBLP:conf/podc/KuhnWZZ03} have shown that any online routing algorithm that only has local knowledge about the network cannot be $c$-competitive.
Based on these results, the authors in \cite{algosensorsPaper} proposed a strategy that makes use of a Hybrid Communication Network to obtain information about location and shapes of holes.
They have proven that in case the convex hulls of radio holes do not intersect, their approach finds $c$-competitive paths in the ad hoc network.

In this paper, we aim for a reduction of the number of cellular infrastructure nodes that have to be considered for the computation of $c$-competitive paths.
To do so, we replace the computation of convex hulls of holes by the computation of bounding boxes.
In addition to \cite{algosensorsPaper}, we also propose a strategy for intersecting bounding boxes.

\subsection{Our Contributions}
We consider any hybrid graph $G = (V,E,E_{AH})$ where the Unit Disk Graph of $V$ is connected. Let $H$ be the set of radio holes in $G$ and $P(h)$ denotes the length of the perimeter of a radio hole $h \in H$. 
For every radio hole, the nodes with maximal/minimal $x$- and $y$-coordinates are called \emph{extreme points}.
Our main contribution is:

\begin{theorem} \label{theorem:mainTheorem}
	For any distribution of the nodes in $V$ that ensures that UDG($V$) is connected and of bounded degree, where the bounding boxes of the radio holes do not overlap, our algorithm computes an abstraction of UDG($V$) in $\mathcal{O}(\log^2 n)$ communication rounds using only polylogarithmic communication work at each node so that $18.55$-competitive paths between all source-destination pairs outside of bounding boxes can be found in an online fashion.
	
	The storage needed by the four extreme points of each radio hole is $\mathcal{O}(|H|)$. 
	For every other node, the space requirement is constant.
\end{theorem}

Note that we do not consider source or destination nodes inside of bounding boxes in this work.
Nevertheless, these can be efficiently handled by a straightforward extension of the protocols presented in this paper by the routing algorithm of \cite{algosensorsPaper} for the case that at the source, the target or both nodes lie inside of a hole abstraction.
We also consider intersecting bounding boxes. For pairwise intersecting bounding boxes, we prove the following
\begin{theorem} \label{theorem:mainTheorem2}
	For any distribution of the nodes in $V$ that ensures that UDG($V$) is connected and of bounded degree, where the bounding boxes of the radio holes do pairwise overlap and the convex hulls of holes do not overlap, our algorithm computes an abstraction of UDG($V$) in $\mathcal{O}(\log^2 n)$ communication rounds using only polylogarithmic communication work at each node so that $28.83$-competitive paths between all source-destination pairs outside of bounding boxes can be found in an online fashion.
	
	The storage needed by the four extreme points of each radio hole is $\mathcal{O}(|H|)$. 
	For every other node, the space requirement is constant.
\end{theorem}
For multiple bounding box intersections, we prove:
We prove that in case we can find a $c$-competitive path between outer intersection points of bounding boxes, we can also find a $\left(10.68 + c \cdot \onlineRoutingConstantIntersecting\right)$-competitive path between all source-destination outside of bounding boxes.
Since the computation of $c$-competitive paths between outer intersection points is a hard problem, we provide a heuristic solution in this paper.
We show via simulations that our approach outperforms classical online routing strategies for ad hoc network with holes significantly, both for intersecting and non-intersecting bounding boxes.
\subsection{Related Work}
In the context of geometric routing in ad hoc networks, several routing techniques have been investigated.
One of the early approaches is GPSR \cite{gpsr}, in which greedy routing is used whenever possible.
In case a packet reaches a dead end, the packet is routed along the perimeter of the hole via the right-hand rule.
As soon as greedy routing is applicable again, the routing mode is changed to greedy routing.
A similar approach is Compass Routing \cite{compass}.
The algorithm considers the direct line segment connecting the source node $s$ and the target node $t$.
At every step the edge with smallest slope to the direct line segment is chosen.
This, however, does not lead to a delivery guarantee in all kinds of graphs.
An example for a graph with delivery guarantee is the Delaunay Graph which in addition is a $1.998$-spanner of the Euclidean metric~\cite{xiaDelaunaySpanner}.
The value for $c$ is $3.56$.
MixedChordArc is the latest $c$-competitive routing strategy for Delaunay Graphs which has been recently published by Bonichon et al.\ \cite{DBLP:conf/esa/BonichonBCDHS18}.
The authors in \cite{compass} introduce a strategy that combines compass routing with face routing to obtain a routing strategy with delivery guarantee for all kinds of connected geometric graphs.
Several extensions of these original ideas have been investigated.
Some of these extensions are FACE-I, FACE-II, AFR, OAFR, GOAFR and GOAFR+ \cite{face12,DBLP:conf/dialm/KuhnWZ02,DBLP:conf/mobihoc/KuhnWZ03,DBLP:conf/podc/KuhnWZZ03}.
In \cite{DBLP:conf/mobihoc/KuhnWZ03,DBLP:conf/podc/KuhnWZZ03} it is proven that the strategies GOAFR and GOAFR+ are asymptotically optimal.
GOAFR and GOAFR+ achieve path length which have a quadratic competitiveness compared to the shortest path.
INF \cite{randomForwarding} is an approach that combines greedy forwarding with randomness.
In case a packet gets stuck via greedy routing, a random intermediate location is chosen.
This, however, requires some global knowledge to choose a random node which is not too far away from the location.
Additionally, INF does not ensure delivery guarantee. 

In \cite{yu2001geographical}, also greedy routing with a modification of face routing is used.
To overcome potential bottlenecks which avoid a guaranteed delivery, a restricted flooding procedure is used. 
In a slightly different setting, namely nodes on the grid, a packet is routed along multiple paths and is hence comparable to a restricted flooding procedure \cite{ruhrup2006online}.
In their model, alive node and crashed nodes exist on the grid. 
The crashed nodes behave like obstacles on the grid which have to be avoided by routing paths. 

In addition to the just mentioned local routing, there are also routing strategies that use a portion of global knowledge about the network.
BoundHole~\cite{DBLP:conf/infocom/FangGG04}, for instance, uses a preprocessing phase at each node which is located at the boundary of a hole.
These hole nodes send out a packet which is routed using the right hand rule around the perimeter of the hole until it reaches the source of the message.
On the way, the packet collects information about the boundary of the hole.
With knowledge about the boundary, the authors are able to find better paths than strategies which only use local~information. 
For a survey on all mentioned strategies, we refer the reader to \cite{Ahmed:2005:HPW:1072989.1072992}.\\

To combine local and global routing strategies, where the goal is to use only few global knowledge, Hybrid Communication Networks have been introduced \cite{algosensorsPaper}.
Hybrid Communication Networks have also been proposed in different contexts.
In practical applications, the term Hybrid Communication Network usually combines wired with wireless networks like in \cite{710380,hybridWiredWireless}.
Closer to our application is the scenario presented in \cite{DBLP:conf/icalp/GmyrHSS17}.
The authors assume an external network which is not under control of the network participants.
The participants can, however, control an internal network.
The authors show that the combination of both networks allows to evaluate monitoring problems of the external network much faster than in classical approaches which only use
the links of the external network. 

The approach we extend in this work makes use of global information as well \cite{algosensorsPaper}.
The global information is gathered via a Hybrid Communication Network.
In a Hybrid Communication Network, nodes can communicate with other nodes in their ad hoc range for free.
In addition, they can use long-range links to communicate with any other node of the network.
These long-range links, the Cellular Infrastructure, are costly.
The solution they propose is to compute an Overlay Network in which holes are represented by their convex hulls.
It is assumed that the convex hulls of the holes do not intersect.
The storage requirements for some nodes are asymptotically in the size of the sum of all holes.
In this work, we aim to reduce the storage requirements for these nodes and investigate also the challenging question of $c$-competitive routing through intersections of hole abstractions.\\

\section{Preliminaries} \label{section:preliminaries}
Initially, we define our ad hoc network topology and provide some general results about routing in the ad hoc network.
\subsection{Properties of the ad hoc network}
In this paper, we assume that the nodes of the ad hoc network are in general position, i.e., there are no three nodes on a line and no four nodes on a cycle.
Moreover, we assume that the coordinates of each node are unique and thus there are no two nodes on the same position. We consider a \twoDel{} as topology for the ad hoc network which is related to the Delaunay Graph.
Let $\bigcirc\left(u,v,w\right)$ be the unique circle through the nodes $u, v$ and $w$ and $\triangle \left(u,v,w\right)$ be the triangle formed by the nodes $u,v$ and $w$.
For any $V\subseteq \mathbb{R}^2$, the \emph{Delaunay Graph} $\mathrm{Del}\left(V\right)$ of $V$ contains all triangles $\triangle \left(u,v,w\right)$ for which $\bigcirc\left(u,v,w\right)$ does not contain any further node besides $u,v$ and $w$.
The $2$-localized Delaunay Graph is a structure that only allows edges which do not exceed the transmission range of a node. 
In $k$-localized Delaunay Graphs, a triangle $\triangle \left(u,v,w\right)$ for nodes $u,v,w$ of $V$ satisfies that all edges of $\triangle \left(u,v,w\right)$ have length at most $1$ and the interior of the disk $\bigcirc \left(u,v,w\right)$ does not contain any node which can be reached within $k$ hops from $u,v$ or $w$ in UDG($V$). 
The \emph{$k$-localized Delaunay Graph} $LDel^k\left(V\right)$ is defined to consist of all edges of $k$-localized triangles and all edges $\left(u,v\right)$ for which the circle with diameter $\overline{uv}$ does not contain any further node $w \in V$.
For $k=2$, we obtain the $2$-localized Delaunay Graph which is also a planar graph~\cite{localDelaunay}.
\twoDel s can be constructed in a constant number of communication rounds \cite{algosensorsPaper}.
Since $2$-localized Delaunay Graphs do not contain all edges of a corresponding Delaunay Graph, one cannot simply use routing strategies for Delaunay Graphs in our scenario.
We denote faces of the $2$-localized Delaunay Graph which are not triangles as \emph{holes}.
For the formal definition of holes, we distinguish between \emph{inner} and \emph{outer} holes.
The definition of inner holes is similar to the definition used in~\cite{DBLP:conf/infocom/FangGG04}.
\begin{definition} [Hole] \label{definition:innerHole}
	Let $V \in \mathbb{R}^2$.
	An \emph{inner hole}  is a face of $LDel^2\left(V\right)$ with at least $4$ nodes.
	Furthermore, let $CH\left(V\right)$ be the set of all edges of the convex hull of $V$.
	Define $\overline{LDel^2}\left(V\right)$ to be the graph that contains all edges of the \twoDel{} and $CH\left(V\right)$. 
	An \emph{outer hole} is a face in $\overline{LDel^2}\left(V\right)$ with at least $3$ nodes, that contains an edge $e \in CH\left(V\right)$ with $\|e\| > 1$.
\end{definition}
Nodes lying on the perimeter of a hole are called \emph{hole nodes}.
Note that the hole nodes of the same hole form a ring, i.e., each hole node  is adjacent to exactly two other hole nodes for each hole it is part of.
The choice of the $2$-localized Delaunay Graph as network topology is motivated by its \emph{spanner}-property.
The Delaunay Graph $\mathrm{Del}\left(V\right)$ contains paths between every pair of nodes $v$ and $w$ of $V$ which are not longer than $c$ times their Euclidean distance.
Delaunay Graphs are proven to be geometric $1.998$-spanners~\cite{xiaDelaunaySpanner}.
Xia argues that the bound of $1.998$ also relates to $2$-localized Delaunay Graphs~\cite{xiaDelaunaySpanner}.
However, these graphs are not spanners of the Euclidean metric but of the Unit Disk Graph.
For the ease of notation, whenever we say that there is a $c$-competitive path in the \twoDel{} we mean that the path has length at most $c$ times the length of the shortest possible path in the Unit Disk Graph of the same node set.
\subsection{Competitive Routing in \twoDel{}s}
In general, we cannot apply routing strategies for the Delaunay Graph in \twoDel{}s since \twoDel{}s contain holes.
In this section, however, we prove that \twoDel{}s and Delaunay Graphs do not differ in dense regions and hence we can apply routing strategies for the Delaunay Graph
between visible nodes, i.e., pairs of nodes which direct line segment does not intersect any hole.
\skipSpace
\begin{theorem} \label{theorem:localDelaunayVisibilityPath}
	Let $G_{2Del} = (V,E_{2Del})$ be a \twoDel{} and $s,t \in V$ such that the line segment $\overline{st}$ does not intersect any hole of $G_{2Del}$. 
	Then, there exists a path $p$ between $s$ and $t$ in $G_{2Del}$ such that
	\begin{center}
		$\eucl{p} \leq 1.998 \cdot \eucl{st}.$
	\end{center}
\end{theorem}

To prove \Cref{theorem:localDelaunayVisibilityPath}, we make use of a definition which was introduced by Bose et al.\ \cite{competitiveRoutingInDelaunayImproved}. 
Let $s$ and $t$ be nodes of a Delaunay Graph.
Bose et al.\ considered the chain of triangles intersected by the line segment $\overline{st}$. 
Each of these triangles contains an edge which either lies completely above or below $\overline{st}$. 
If we consider only these edges, we can see that these edges form a polygon.
Walking along all edges lying above $\overline{st}$ describes a path between $s$ and $t$.
This path is called \textit{upper chain} of $s$ and $t$ ($\mathrm{UC}(s,t)$) and the corresponding path for all edges which lie below $\overline{st}$ is called \textit{lower chain} of $s$ and $t$ ($\mathrm{LC}(s,t)$).
Xia has proven that between any pair of nodes $s$ and $t$ in a Delaunay Graph a path with length at most $1.998 \cdot \eucl{st}$ exists \cite{xiaDelaunaySpanner}.
The path construction of Xia uses only edges which connect nodes of $\mathrm{UC}(s,t)$ and $\mathrm{LC}(s,t)$.
We use this knowledge and show that in Delaunay Graphs a polygon described by an upper and a lower chain of nodes $s$ and $t$ never contains any edge with a length larger than $1$, provided $s$ and $t$ are visible from each other in the corresponding \twoDel.
Afterward, we conclude that between any pair of visible nodes $s$ and $t$ in a \twoDel{} a path with length at most $1.998 \cdot \eucl{st}$ exists.
\skipSpace
\begin{lem} \label{lemma:visibilityPath}
	Given a \twoDel{} $G_{2Del}~=~(V,E_{2Del})$ and two nodes $s$ and $t$ such that the line segment $\overline{st}$ does not intersect any hole of $G_{2Del}$.
	Let $G_{Del}~=~(V,E_{Del})$ be the Delaunay Graph to the same point set $V$.
	The polygon described by $\mathrm{UC}(s,t)$ and $\mathrm{LC}(s,t)$ in $G_{Del}$ does not contain any edge $e$ with $\eucl{e} >1$.
\end{lem}

\begin{proof}
	There are three types of edges which are part of the polygon with boundaries $\mathrm{UC}(s,t)$ and $\mathrm{LC}(s,t)$ in $G_{Del}$.
	Edges that cross the line segment $\overline{st}$, edges which lie completely above $\overline{st}$ and edges that lie completely below $\overline{st}$.
	We prove for every type of edges that these cannot be larger than one in case $s$ and $t$ are visible from each other in $G_{2Del}$. \\
	\skipSpace
	\textbf{Case 1:} Edges crossing $\overline{st}$: \\
	\skipSpace
	Without loss of generality, let $\triangle abc$ be a triangle which is intersected by $\overline{st}$ and $\overline{ab}$ an edge that crosses $\overline{st}$.
	Assume $\eucl{ab} > 1$.
	This immediately implies that $\overline{st}$ crosses a hole since $\overline{st}$ intersects a face with at least $4$ nodes which is a contradiction to our assumption.  \\
	\skipSpace
	\textbf{Case 2:} Edges above $\overline{st}$: \\
	\skipSpace
	Without loss of generality, let $\triangle abc$ be a triangle which is intersected by $\overline{st}$ and $\overline{ab}$ an edge that lies above $\overline{st}$ with $\eucl{ab} > 1$.
	With this knowledge we can conclude that $\overline{ac\vphantom{b}}$ and $\overline{bc}$ potentially lie on the perimeter of a hole if $\overline{ac\vphantom{b}}$ and $\overline{bc}$ are not hole edges themselves. 
	Again we can easily see that $\overline{st}$ would cross a hole in this case which is a contradiction to our assumption. \\
	\skipSpace
	\textbf{Case 3:} Edges below $\overline{st}$: \\
	\skipSpace
	We can apply the same argumentation as for Case $2$ here. \\
	\skipSpace
	Since every possible type of edges cannot have a length larger than $1$, we have proven \Cref{lemma:visibilityPath}.

\end{proof}
\skipSpace
\Cref{lemma:visibilityPath} implies that we can apply routing strategies for Delaunay Graphs also between visible nodes in \twoDel{}s.
This leads to the relation between our routing strategy and Visibility Graphs.
In the Visibility Graph $Vis\left(V\right)$ of a set of polygons, $V$ represents the set of corners of the polygons, and there is an edge $\{v,w\}$ in $Vis\left(V\right)$ if and only if a line can be drawn from $v$ to $w$ without crossing any polygon, i.e., $v$ is visible from $w$.
De Berg et al.\ showed that it is enough to consider nodes of obstacle polygons for finding shortest paths in polygonal domains~\cite{computationalGeometryDeBerg}.
Hence, if we consider the Visibility Graph of holes of the \twoDel{}, we can translate a path in the Visibility Graph to a path in \twoDel{} by applying a routing strategy for Delaunay Graphs along every edge on the path in the Visibility Graph.
As we do not want to store large routing tables, we are interested in online routing strategies for the Delaunay Graph.
In this work, we make use of the online strategy MixedChordArc \cite{DBLP:conf/esa/BonichonBCDHS18} which finds $\mixedChord$-competitive paths between every source and target node in the Delaunay Graph. 
To sum it up, knowledge about the Visibility Graph of holes enables us to find $\mixedChord$-competitive paths in the \twoDel{} between any pair of nodes $(s,t)$ by applying the MixedChordArc-strategy along every edge of the shortest path between $s$ and $t$ in the Visibility Graph.
\section{Geometric properties of bounding box paths} \label{section:geometry}
We have seen that knowledge about the Visibility Graph of holes enables us to find $c$-competitive paths in the \twoDel{}.
In general, however, the node set of a Visibility Graph can be very large since potentially many nodes could lie on the boundary of holes.
To reduce the space constraints and to speed up the computation of $c$-competitive paths, we aim for a reduction of the number of nodes in the Visibility Graph while still being able to find $c$-competitive paths.
To do so, we reduce the Visibility Graph by only considering the bounding boxes of holes.
The following definition defines the axis-parallel bounding box of a hole.
\skipSpace
\begin{definition}[Bounding Box] \label{definition:boundingBox}
	\skipSpace
	Let $p$ be a polygon. 
	Let $min_x = \min_{v \in p} (x(v))$ and $max_x, min_y$ and $max_y$ be defined analogously. 
	These points are called \emph{extreme points} of $p$.
	The (axis-parallel) \emph{Bounding Box} of $p$ is a polygon with the following nodes: 
	\begin{enumerate}
		\item $bb_{t\ell}(p) = (x(min_x),y(max_y))$ (top-left)
		\item $bb_{tr}(p) = (x(max_x),y(max_y))$ (top-right)
		\item $bb_{b\ell}(p) = (x(min_x),y(min_y))$ (bottom-left)
		\item $bb_{br}(p) = (x(max_x),y(min_y))$ (bottom-right)
	\end{enumerate}
	The nodes are connected via the direct line segments $\overline{bb_{t\ell}(p)bb_{tr}(p)}, \overline{bb_{tr}(p),bb_{br}(p)},$ $\overline{bb_{t\ell}(p)bb_{b\ell}(p)}$ and $\overline{bb_{b\ell}(p)bb_{br}(p)}$.
\end{definition}
In the following, we see how we can embed bounding boxes of holes in the \twoDel{} and that considering only bounding boxes of holes allows us to find $c$-competitive paths between every source
and target node that lies outside of any bounding box.

\subsection{Embedding of Bounding Boxes} \label{section:embedding}
In general, nodes of bounding boxes of holes do not match with any nodes of the ad hoc network (see \Cref{fig:bbLocal,fig:bbNotLocal}).
In this section, we propose an embedding of bounding boxes in the \twoDel{} and prove later on that we can find $c$-competitive paths with help of the embedding.
Since we consider a given \twoDel{} with node set $V$ and edge set $E$, we have to find nodes in $V$ that represent nodes of bounding boxes.
Nodes of bounding boxes of holes are called \textit{real} bounding box nodes whereas nodes of $V$ that represent real bounding box nodes are denoted as \textit{representatives} of a real bounding box.

\begin{figure}[htbp]
	\begin{minipage}{0.49\textwidth} 
		\includegraphics[width=\textwidth]{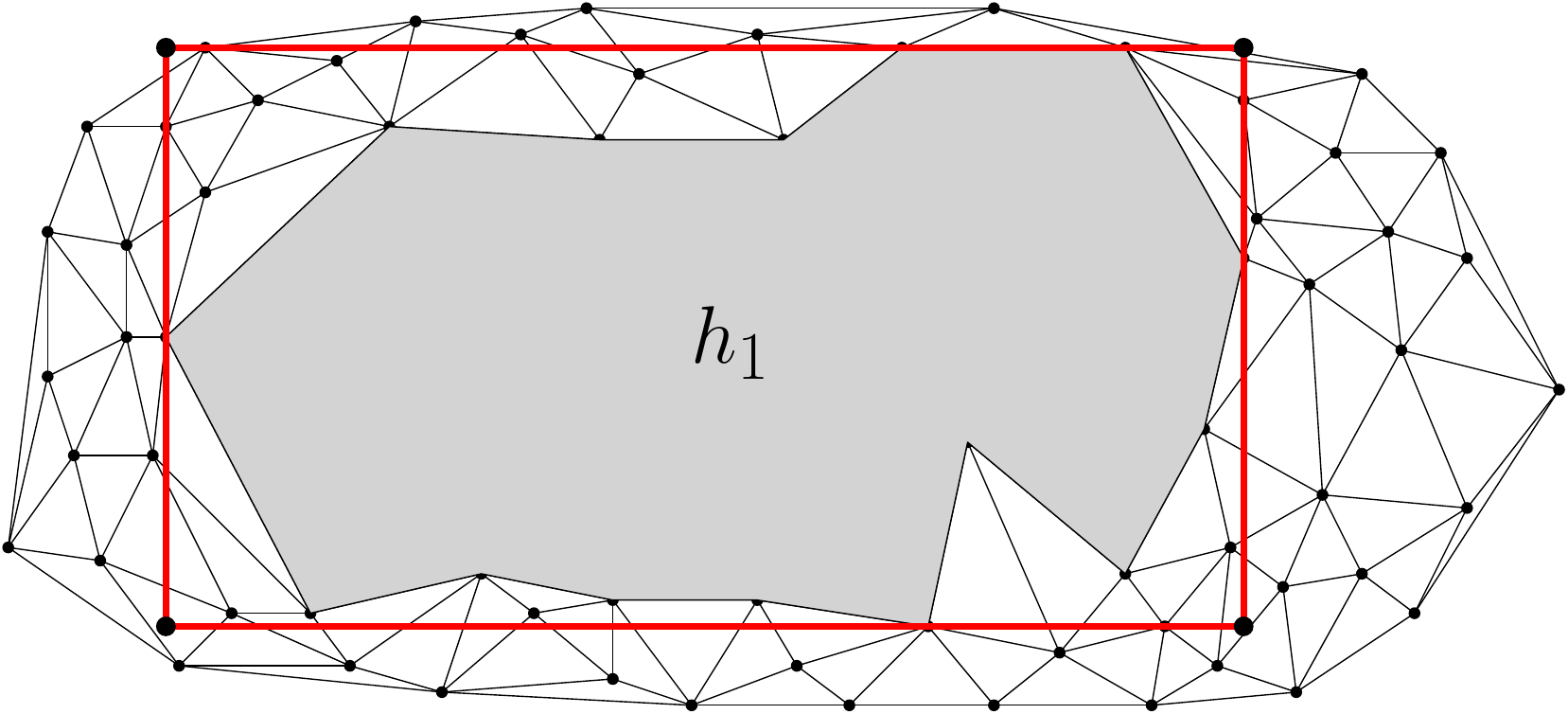}
		\caption[A bounding box in a \twoDel{}.]{The nodes of the red bounding box are not part of the \twoDel.}
		\label{fig:bbNotLocal} 
	\end{minipage}
	\hfill
	\begin{minipage}{0.49\textwidth}
		\includegraphics[width=\textwidth]{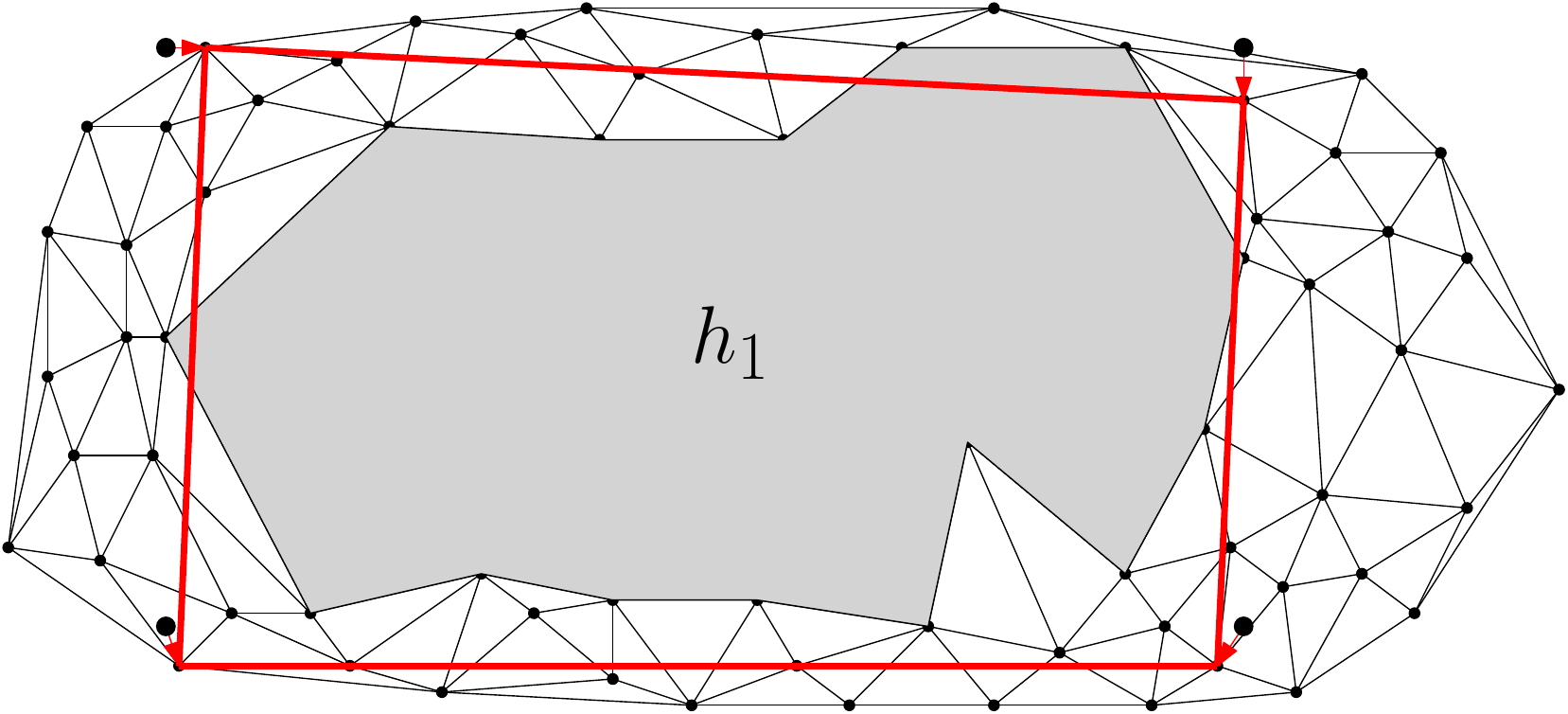}
		\caption[A bounding box of representatives in a \twoDel{}.]{A bounding box described by its representatives of the ad hoc network.}
		\label{fig:bbLocal} 
	\end{minipage}
\end{figure}
\noindent
The solution is to choose those nodes of $V$ as representatives of real bounding box nodes which have the shortest distance to a real bounding box node.
More formally, let $G_{2Del} = (V,E)$ be a \twoDel{} with corresponding Voronoi Diagram $\mathrm{Vor}(V)$. 
A real bounding box node $b$ is represented by the node of the Voronoi Cell $c \in \mathrm{Vor}(V)$ with $b \in c$.
The resulting bounding box (see \Cref{fig:bbLocal}) does not necessarily enclose the entire hole anymore but we prove that it has similar properties as a real bounding box.
Since we have proven that $c$-competitive paths between visible nodes in \twoDel{} exist, our idea is to use the direct line between real bounding box nodes for routing decisions.
We call this direct line \textit{virtual Axis}.
\skipSpace
\begin{definition}[Virtual Axis]
	\skipSpace
	Consider a \twoDel{} $G_{2Del}= (V,E)$ with nodes $s,t \in V$.
	Let $C_s$ and $C_t$ be the cells of the corresponding Voronoi Diagram with $s \in C_s$ and $t \in C_t$. 
	Additionally let $a,b \in \mathbb{R}^2$ with $a \in C_s$ and $b \in C_t$ but $a,b \notin V$ and $a \neq b$.
	We call the line segment $\overline{ab}$ a \emph{virtual Axis} between $s$ and $t$ in $G_{2Del}$.
	For the ease of notation, we simply write ${\mathrm{vAxis}(s,t)}$.
\end{definition}
\skipSpace
In our scenario, we use a virtual Axis between visible real bounding box nodes. 

\noindent
After clarifying the definition of virtual Axes, we can introduce the main theorem of this section. 
We prove that there exists a path with length at most $3.996$ times the Euclidean distance between the real bounding box nodes between two nodes $s$ and $t$ representing two adjacent real bounding box nodes. 
\skipSpace

\begin{theorem} \label{theorem:virtualAxisPaths}
	Let $G_{2Del} = (V,E)$ be a \twoDel{} with $s,t \in V$.
	For any $\mathrm{vAxis}(s,t)$ with endpoints $bb_{t\ell}$ and $bb_{tr}$ that does not intersect any hole of $G_{2Del}$, there exists a path $p$ between $s$ and $t$ in $G_{2Del}$ with length at most:
	\begin{center}
	$\eucl{p} \leq 3.996 \cdot \eucl{bb_{t\ell}bb_{tr}}.$
	\end{center}
\end{theorem}
\skipSpace
We split the proof of Theorem \ref{theorem:virtualAxisPaths} into two lemmas. \Cref{lemma:lengthOfST} bounds the length of the line segment $\overline{st}$ and \Cref{lemma:shortestPolylineST} proves that $\mathrm{vAxis(s,t)}$ is a candidate for the shortest polyline (see \Cref{definition:shortestPolyline}) between $s$ and $t$. The combination of both yields \Cref{theorem:virtualAxisPaths}.
\skipSpace
\begin{lem} \label{lemma:lengthOfST}
	Let $G_{2Del} = (V,E,w)$ be a \twoDel{} with a pair of nodes $s,t \in V$ and let $\mathrm{vAxis}(s,t)$ be a virtual Axis between $s$ and $t$ that does not intersect any hole of $G_{2Del}$.
	The endpoints of $\mathrm{vAxis}(s,t)$ are denoted as $bb_{t\ell}$ and $bb_{tr}$.
	Then, 
	\begin{center}
	$\eucl{st} \leq 2 \cdot \eucl{bb_{t\ell}bb_{tr}}.$
	\end{center}
\end{lem}
\skipSpace
\begin{proof}
	By the triangle inequality $\eucl{st} \leq \eucl{sbb_{t\ell}} + \eucl{bb_{t\ell}bb_{tr}} + \eucl{bb_{tr}t}$.
	Next, we bound the length of the line segments $\overline{sbb_{t\ell}}$ and $\overline{bb_{tr}t}$.
	\noindent
	We argue that the length of each line segment is at most $\frac{1}{2}$.
	$s$ is the representative of $bb_{t\ell}$ since $s$ is the node with smallest distance to $bb_{t\ell}$ of all nodes in $V$.
	$bb_{t\ell}$ lies either in or on the boundary of a triangle $t_s$ that has $s$ as a node. 
	Each edge of this triangle has length at most $1$ due to the properties of \twoDel{}s.
	When considering triangles, it is easy to see that the endpoints of a triangle -- $t_s$ in our case -- are those nodes which have the largest distances to each other in the triangle. 
	Consequently, the worst possible case is that $bb_{t\ell}$ falls exactly on the half of an edge with length $1$. 
	In this case the closest point is at distance $\frac{1}{2}$ which is an upper bound for $\eucl{bb_{t\ell}s}$.
	We can use the same arguments for the line segment $\overline{bb_{tr}t}$.
	
	\noindent
	Thus, we can bound the length of $\overline{st}$ as follows:
	\begin{align*}
	&\phantom{what} & \eucl{st} &\leq \eucl{sbb_{t\ell}} + \eucl{bb_{t\ell}bb_{tr}} + \eucl{bb_{tr}t} \\
	&&&\leq \frac{1}{2} + \eucl{bb_{t\ell}bb_{tr}} + \frac{1}{2}\\
	&&&= \eucl{bb_{t\ell}bb_{tr}} +1
	\end{align*}
	Further, $\eucl{bb_{t\ell}bb_{tr}} > 1$, due to the definition of holes. 
	Thus we obtain a final bound on $\eucl{st}$:
	\begin{align*}
	& \phantom{what} &\eucl{st} &\leq  \eucl{bb_{t\ell}bb_{tr}} +1 \\
	&&&\leq \eucl{bb_{t\ell}bb_{tr}} + \eucl{bb_{t\ell}bb_{tr}} \\
	&&&= 2 \cdot \eucl{bb_{t\ell}bb_{tr}}.
	\end{align*}
\end{proof}
\noindent	
After being able to express the length of $\overline{st}$ in terms of $\eucl{bb_{t\ell}bb_{tr}}$, we start with proving that a $c$-competitive path between $s$ and $t$ exists.
The proof is inspired by the path construction for Delaunay Graphs introduced by Xia.
We need two definitions which have been introduced by Xia \cite{xiaDelaunaySpanner}.
\skipSpace
\begin{defi}[Chain of Disks]
	\skipSpace
	A finite sequence of disks $\mathcal{O} = (O_1, O_2, \dots, O_n)$ is called \emph{chain of disks} if it has the following two properties: 
	
	\noindent
	\textbf{Property 1:} \\
	Every pair of consecutive disks $O_i$ and $O_{i+1}$ intersects but neither disk contains the other. 
	
	\noindent
	Denote by $C_i^{(i-1)}$ and $C_i^{(i+1)}$ the arcs on the boundary of $O_i$ that are in $O_{i-1}$ and $O_{i+1}$ respectively. 
	These arcs are denoted as \emph{connecting arcs} of $O_i$. 
	
	\noindent
	\textbf{Property 2:} \\
	The connecting arcs of $O_i$ do not overlap for $2 \leq i \leq n-1$, however they can share an endpoint.
	
	\noindent
	Two points $u$ and $v$ are called \emph{terminals} of $\mathcal{O}$ if $u$ lies on the boundary of $O_1$ and is not in the interior of $O_2$ and $v$ lies on the boundary of $O_n$ and is not in the interior of $O_{n-1}$.
\end{defi}
\skipSpace
\begin{defi}[Shortest polyline between $u$ and $v$] \label{definition:shortestPolyline}
	\skipSpace
	Given a chain of disks $\mathcal{O} = (O_1,O_2, \dots, O_n)$ with terminals $u$ and $v$.
	Let $o_1, \dots , o_n$ be the centers of $O_1, \dots, O_n$.
	The polyline $uo_1\dots o_nv$ is called the \emph{centered polyline} between $u$ and $v$.
	For $1 \leq i \leq n-1$, let $a_i$ and $b_i$ be the intersections of the boundaries of $O_i$ and $O_{i+1}$.
	Without loss of generality, all $a_{i}$'s are assumed to be on one side of the centered polyline and all $b_i$'s are on the other side. 
	For notational convenience, define $a_0 = b_0 = u$ and $a_n = b_n = v$.
	Let $D_{\mathcal{O}}(u,v) = up_1 \dots p_{n-1}v$ be the \emph{shortest polyline} from $u$ to $v$ that consists of line segments $\overline{up_1\vphantom{v}}, \overline{p_1p_2}, \dots \overline{p_{n-1}v}$ where $p_i \in \overline{a_ib_i}$ for $1 \leq i \leq n-1$.
\end{defi}
\noindent

\begin{figure}[h]
	\centering
	\includegraphics[width=0.75\textwidth]{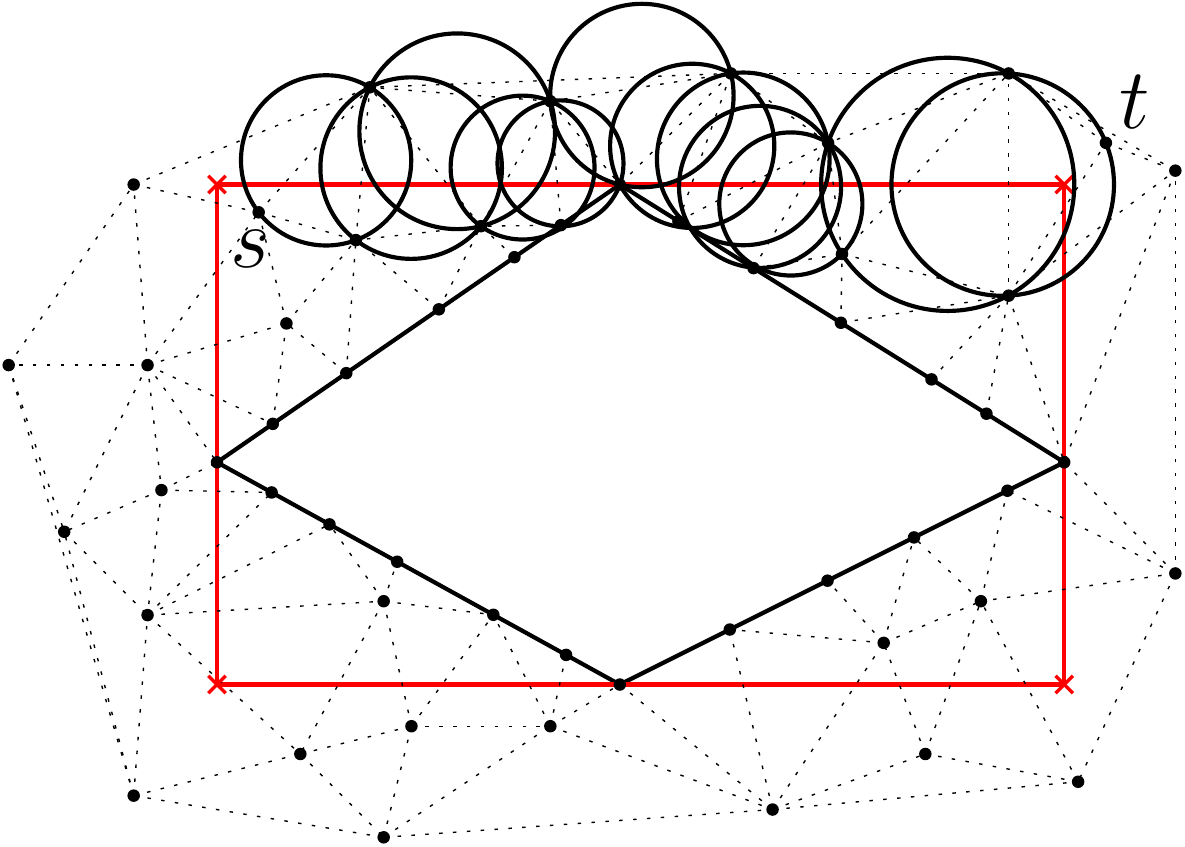}
	\caption{The Chain of Disks $\mathcal{O}$ from $s$ to $t$ along $\mathrm{vAxis}(s,t)$.}.
	\label{fig:diskChainst} 
\end{figure}
\noindent
With these definitions, we can state Lemma \ref{lemma:shortestPolylineST}.
\skipSpace
\begin{lem} \label{lemma:shortestPolylineST}
	Let $G_{2Del} = (V,E,w)$ be a \twoDel{} with a pair of nodes $s,t \in V$ and let $\mathrm{vAxis}(s,t)$ be a virtual Axis that does not intersect any hole of $G_{2Del}$
	with endpoints $bb_{t\ell}$ and $bb_{tr}$.
	Further let $\mathcal{O}$ be the chain of disks with terminals $s$ and $t$ obtained by the circumcircles of all triangles intersected by $\mathrm{vAxis}(s,t)$.
	Then, we can bound the shortest polyline  $D_{\mathcal{O}}(s,t)$ as follows:
	\begin{center}
	$D_{\mathcal{O}}(s,t) \leq 2 \cdot \eucl{bb_{t\ell}bb_{tr}}.$
	\end{center}
\end{lem}

\begin{proof}
	In \cite{xiaDelaunaySpanner}, the author argues that any sequence of disks obtained by the circumcircles of triangles along a line segment in a Delaunay Graph is a chain of disks.
	Due to Lemma \ref{lemma:visibilityPath} we know that $\mathrm{vAxis}(s,t) = \overline{bb_{t\ell}bb_{tr}}$ does not intersect any hole.
	Consider the chain of disks $\mathcal{O}$ with terminals $s$ and $t$ obtained by the circumcircles of all triangles intersected by $\mathrm{vAxis}(s,t)$.
	See Figure \ref{fig:diskChainst} for a visualization.
	The main observation for our proof is that the polyline $sbb_{t\ell}bb_{tr}t$ fulfills the requirements of \Cref{definition:shortestPolyline} and is a candidate for the shortest polyline between $s$ and $t$.
	Thus, $sbb_{t\ell}bb_{tr}t$ is an upper bound for $D_{\mathcal{O}}(s,t)$.
	Hence:
	\begin{align*}
	&\phantom{space}&D_{\mathcal{O}}(s,t) &\overset{\phantom{Lemma 1}}{\leq} \eucl{sbb_{t_\ell}} + \eucl{bb_{t\ell}bb_{tr}} + \eucl{bb_{tr}t} \hspace*{2em}\\
	&&&\overset{\Cref{lemma:lengthOfST}}{\leq}\hspace*{0.5em} 2 \cdot \eucl{bb_{t\ell}bb_{tr}}
	\end{align*}

\end{proof}
\noindent	
The combination of \Cref{lemma:lengthOfST} and \Cref{lemma:shortestPolylineST} helps us to prove \Cref{theorem:virtualAxisPaths}. 
Xia states that the shortest connection $P_{\mathcal{O}}(s,t)$ between two terminal nodes $s$ and $t$ along a chain of disks is at most $1.998 \cdot D_{\mathcal{O}}(s,t)$ \cite{xiaDelaunaySpanner}.
Thus, we obtain that there exists a path between $s$ and $t$ with length at most:
\begin{align*}
&\phantom{space}&P_{\mathcal{O}}(s,t) &\overset{\phantom{\Cref{lemma:shortestPolylineST}}}{\leq} 1.998 \cdot D_{\mathcal{O}}(s,t) \hspace*{2em} \\
&&&\overset{\Cref{lemma:shortestPolylineST}}{\leq}  1.998 \cdot 2 \cdot  \eucl{bb_{t\ell}bb_{tr}} = 3.996 \cdot  \eucl{bb_{t\ell}bb_{tr}}.
\end{align*}
\noindent
\skipSpace
So far, we concentrated on proving the existence of such a path.
Nevertheless, we are also able to find a $c$-competitive path via the MixedChordArc-algorithm.
To do so, we slightly modify the algorithm such that we do not use the direct line segment between two representatives as referencing segment but the virtual axis connecting the real bounding box vertices.
The analysis of MixedChordArc \cite{DBLP:conf/esa/BonichonBCDHS18} proves that the path found along the virtual axis has length at most $\mixedChord$ times the length of the virtual axis.
The entire path has length at most $5.56$ times the length of the virtual axis since the connection between $s$ and the first node along the path and $t$ and the last node on the path has length at most $2$ times the length of the virtual axis. 
This leads to the following corollary.
\skipSpace
\begin{corollary} \label{corollary:onlineRouting}
	Let $G_{2Del} = (V,E)$ be a \twoDel{} with $s,t \in V$.
	For any $\mathrm{vAxis}(s,t)$ with endpoints $bb_{t\ell}$ and $bb_{tr}$ that does not intersect any hole of $G_{2Del}$, there exists an online routing strategy that finds a path $p$ between $s$ and $t$ in $G_{2Del}$ with length at most:
	\begin{center}
	$\eucl{p} \leq 5.56 \cdot \eucl{bb_{t\ell}bb_{tr}}.$
	\end{center}
\end{corollary}
\skipSpace
The results of this section enable us to reduce the problem of finding $c$-competitive paths in the \twoDel{} to finding $c$-competitive paths in Visibility Graphs as it is done in the next section.
\subsection{Competitive Paths via non-intersecting Bounding Boxes} \label{section:shortestBBPaths}
Based on our results of \Cref{section:embedding}, we reduce the problem of finding $c$-competitive paths via bounding boxes to finding $c$-competitive paths in Visibility Graphs.
Therefore, we introduce a special class of Visibility Graphs, namely \textit{Bounding Box Visibility Graphs}. 
In Bounding Box Visibility Graphs, each obstacle (holes in our case) is represented by its axis-parallel bounding box. 
Consequently, $V$ consists of the nodes of the axis-parallel bounding box of each obstacle.
Moreover, $E$ consists of the edges of each bounding box as well as of edges between visible nodes of different bounding boxes. 
In this setting, we call two nodes to be visible from each other in case their direct line segment does not intersect any bounding box. \\
Let $O$ be a set of polygonal obstacles, and $s,t \in \mathbb{R}^2$ a source- and a target-location. 
Further let $bb_{t\ell}(p), bb_{tr}(p), bb_{b\ell}(p)$ and $bb_{br}(p)$ be the nodes of an axis-parallel bounding box representing a polygon $p \in O$. 
A Bounding Box Visibility Graph is defined as follows:
\begin{definition}[Bounding Box Visibility Graph]
	A geometric graph $G = (V_{BB},E)$ is called \emph{Bounding Box Visibility Graph} if $bb_{t\ell}(p), bb_{tr}(p), bb_{b\ell}(p)$ and $bb_{br}(p) \in V_{BB}, \forall p \in O$. 
	Additionally, $\set{bb_{t\ell}(p),bb_{tr}(p)},\set{bb_{t\ell}(p),bb_{b\ell}(p)}, \set{bb_{tr}(p),bb_{br}(p)}$ and $\set{bb_{b\ell}(p),bb_{br}(p)} \in E, \forall p \in O$.
	For two nodes of different bounding boxes $u,v \in V_{BB}$, the edge $\{u,v\} \in E$ if $\overline{uv}$ does not intersect the bounding box of any obstacle $p \in O$.
\end{definition}

\begin{figure}[h]
	\centering
	\includegraphics[width=0.9\textwidth]{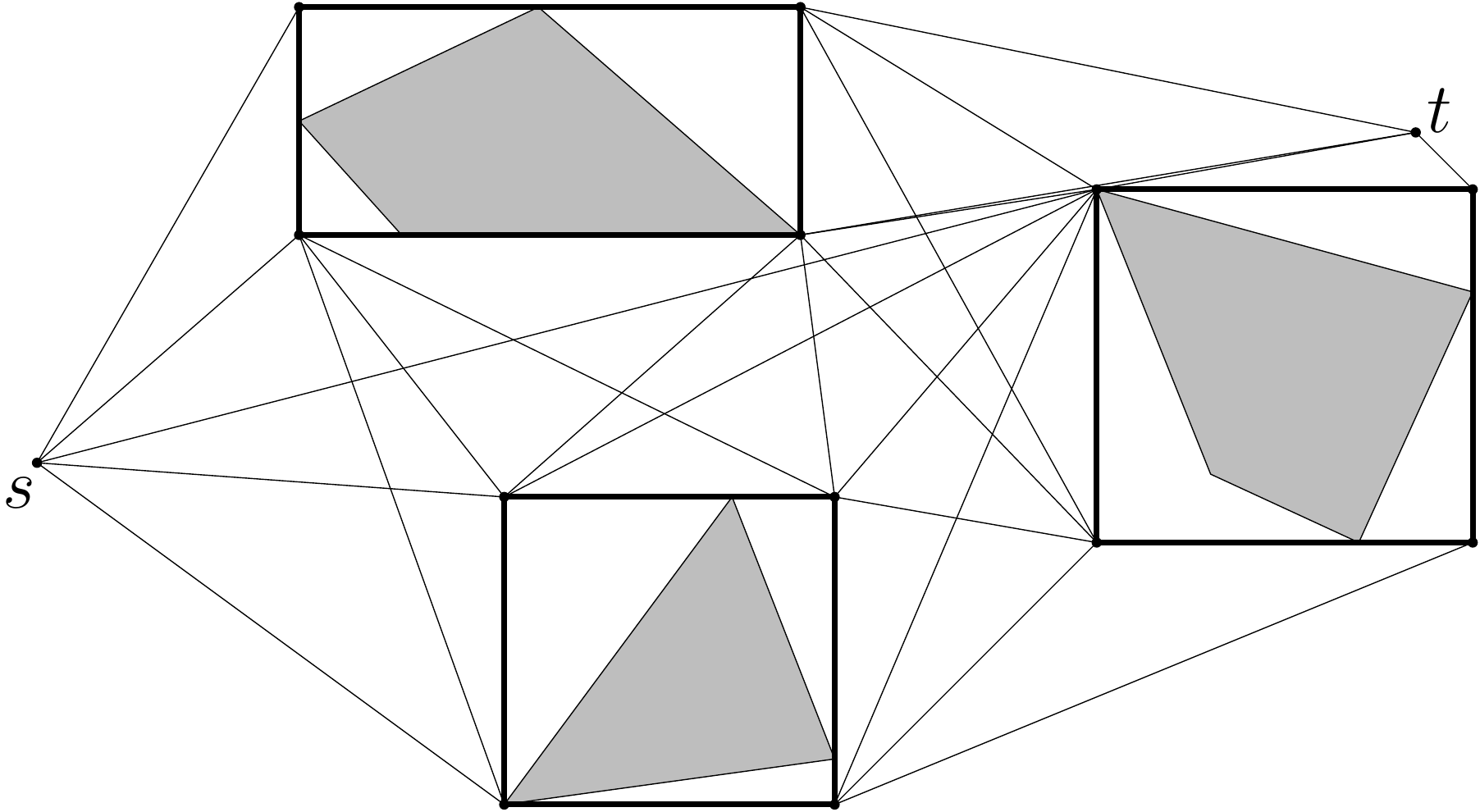}
	\caption[A Bounding Box Visibility Graph for three obstacles.]{A Bounding Box Visibility Graph for three obstacles.}
	\label{fig:visibilityBB} 
\end{figure}

\newpage
\noindent
Figure \ref{fig:visibilityBB} provides a visualization of a Bounding Box Visibility Graph.
To provide an intuition about the proof ideas, we initially assume that the considered Bounding Box Visibility Graph contains only a single bounding box. 
Additionally, we assume that the starting-location $s$ and the target-location $t$ do not lie inside of the bounding box.
The node-set $V$ of the corresponding Bounding Box Visibility Graph contains $s,t$, and the nodes of the bounding box.
More formally: $ V = \{ s,t, bb_{t\ell},bb_{tr},bb_{b\ell},bb_{br} \}$. 
\noindent
The rest of the section deals with proving Theorem~\ref{theorem:singleBB} which states that there always exists 
a path of length at most $\sqrt{2} \cdot d_{\mathrm{UDG}}(s,t)$ between $s$ and $t$ in the described setting. \\

\begin{theorem} \label{theorem:singleBB}
	Let $G = (V,E)$ be a Bounding Box Visibility Graph containing a single bounding box $b$ with a starting-location $s \notin b$ and a target-location $t \notin b$.
	Then, there exists a path between $s$ and $t$ in $G$ with length at most $\sqrt{2}\cdot d_{\mathrm{UDG}}(s,t)$.
\end{theorem}
\skipSpace
For the proof of Theorem \ref{theorem:singleBB} we need a useful property of right triangles. 
Whenever we consider a right triangle with legs $a$ and $b$ and hypotenuse $c$, we can prove that walking along $a$ and $b$ is not longer than a constant times walking along $c$. \Cref{lemma:rightTriangle} deals with this property.
\skipSpace
\begin{lem} \label{lemma:rightTriangle}
	Let $a,b$ be the legs of a right triangle with hypotenuse $c$.  
	Then, $a+b \leq \sqrt{2}\cdot c$.
\end{lem}

\begin{proof}
	The Pythagorean Theorem states $c^2 = a^2 + b^2 \iff c = \sqrt{a^2+b^2}$. 
	We compute the ratio $x$ of $a+b$ and $c$ and bound it afterward. 
	
	\begin{center}
		$\frac{a+b}{c} = \frac{a+b}{\sqrt{a^2+b^2}} \leq x \iff \frac{(a+b)^2}{a^2+b^2} \leq x^2 $
	\end{center}
	
	\noindent
	The equivalence holds as $a,b$, and $c$ are greater than zero.
	After applying the Binomial Theorem, we obtain:
	\begin{center}
	$\frac{(a+b)^2}{a^2+b^2} = \frac{a^2 + 2ab + b^2}{a^2 + b^2} = 
	\frac{a^2 + b^2}{a^2+b^2} + \frac{2ab}{a^2+b^2} = 1 + \frac{2ab}{a^2+b^2}$
	\end{center}
	\noindent
	We analyze the properties of the latter addend and prove:  
	\begin{align*}
	&\phantom{space}&\frac{2ab}{a^2+b^2} \leq 1 &\iff 2ab \leq a^2 + b^2 \iff 0 \leq a^2 - 2ab + b^2 \\
	&&&\iff 0 \leq (a-b)^2
	\end{align*}
	Since quadratic numbers are always positive, our claim holds and we can finally plug all results together and finish the proof.
	\begin{center}
	$\frac{(a+b)^2}{a^2+b^2} = 1 + \frac{2ab}{a^2+b^2} \leq 1 + 1 = 2 = x^2 \iff \sqrt{2} = x$
	\end{center}
\end{proof}

\noindent
With the knowledge of \Cref{lemma:rightTriangle} we are able to prove Theorem \ref{theorem:singleBB}.
\skipSpace
\begin{proof}[Proof of Theorem \ref{theorem:singleBB}]
	Without loss of generality, we assume $x(s) < x(t)$.
	We distinguish two cases concerning the size of the bounding box.
	In Case $1$, we consider bounding boxes that fit completely into the bounding box around $s$ and $t$.
	Case $2$ deals with bounding boxes that exceed the bounding box around $s$ and $t$.
	We consider the surrounding bounding box with nodes $s, t, u = (x(s),y(t))$ and $v = (x(t),y(s))$.
	See Figure \ref{fig:singleBBSmallSurrounded} for a visualization. \\
	\begin{figure}[h]
		\centering
		\includegraphics[height = 0.35 \textheight]{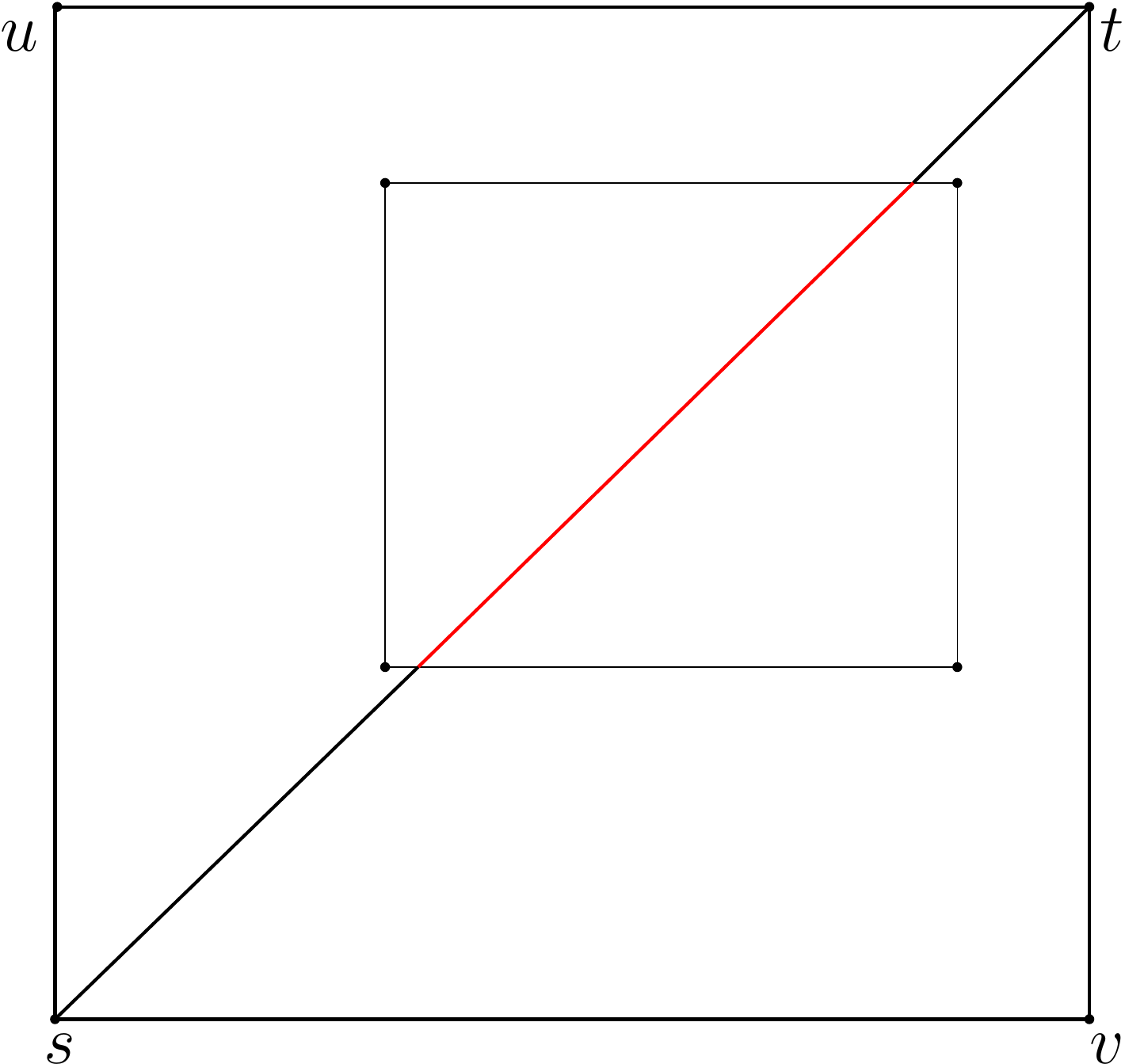}
		\caption{The surrounding bounding box around $s$ and $t$.}
		\label{fig:singleBBSmallSurrounded} 
	\end{figure}
	\skipSpace
	\textbf{Case 1:} The bounding box  of the obstacle is completely cointained in the bounding box with nodes $s,t,u$ and $v$. \\
	\skipSpace
	The worst case is that the bounding box of the hole and the bounding box with nodes $s,t,u$ and $v$ coincide.
	In that case, we see that the shortest connection from $s$ to $t$ is the combination of the line segments $\overline{su\vphantom{t}}$ and $\overline{ut}$. The combination
	of $\overline{sv\vphantom{t}}$ and $\overline{vt}$ has the same length.
	By applying \Cref{lemma:rightTriangle}, we obtain:
	\begin{center}
	$\eucl{su} + \eucl{ut} \leq \sqrt{2} \cdot \eucl{st} \leq \sqrt{2} \cdot d_{\mathrm{UDG}}(s,t).$
	\end{center}
	\skipSpace
	The last inequality holds because the Euclidean distance between $s$ and $t$ is the shortest possible length of a path between $s$ and $t$. 
	The distance of $s$ and $t$ in the Unit Disk Graph has to be larger or equal.
	If the two mentioned bounding boxes do not coincide, the shortest path among nodes of bounding boxes is still smaller than a path along the comprising bounding box.
	Thus, we obtain $\sqrt{2} \cdot \eucl{st}$ as an upper bound for the length of the shortest path among nodes of a bounding box. \\
	\skipSpace
	\textbf{Case 2:} The bounding box of the obstacle exceeds the bounding box with nodes $s,t,u$ and~$v$. \\
	\skipSpace
	In this case, it is enough to consider the cases in which $\overline{ut}$ and $\overline{sv\vphantom{t}}$ or $\overline{us\vphantom{t}}$ and $\overline{tv}$ are intersected by $\overline{st}$.
	Otherwise, we can apply the same argumentation as in Case $1$ since there is a path which is completely contained in the box with nodes $s,t,u$ and $v$.
	Without loss of generality, we assume that both $\overline{ut}$ and $\overline{sv\vphantom{t}}$ are intersected. We can use the same proof for the other case by turning the view about 90 degrees.
	Let $p_{y_{max}}$ be the highest point of the hole polygon and $p_{y_{min}}$ the lowest point respectively. 
	The shortest geometric connection between $s$ and $t$ has to pass either $p_{y_{max}}$ or $p_{y_{min}}$.
	Without loss of generality, we assume that the shortest geometric connection passes $p_{y_{max}}$. 
	The shortest possible (not necessarily realistic) connection would be $\overline{sp_{y_{max}}\vphantom{t}}$ and $\overline{p_{y_{max}}t}$.
	We can upper bound the length of this connection by giving the legs of two right triangles as maximal path length. 
	Consider the points $s_{y_{max}} = (x(s),y(p_{y_{max}}))$ and $t_{y_{max}} = (x(t),y(p_{y_{max}}))$.
	The longest possible path over bounding box points would be $\overline{ss_{y_{max}}\vphantom{t}}$,$\overline{s_{y_{max}}t_{y_{max}}}$ and $\overline{t_{y_{max}}t}$. 
	Since this path uses the legs of right triangles with hypotenuses $\overline{sp_{y_{max}}\vphantom{t}}$ and $\overline{p_{y_{max}}t}$, we can apply \Cref{lemma:rightTriangle} and obtain that the maximal length is at most $\sqrt{2} \cdot (\eucl{sp_{y_{max}}} + \eucl{p_{y_{max}}t}) \leq \sqrt{2} \cdot d_{\mathrm{UDG}}(s,t)$.
	See Figure \ref{fig:singleBBOpp} for a visualization of both right triangles.
	The last inequality holds since the distance between $s$ and $t$ in the Unit Disk Graph cannot be smaller than the shortest possible geometric connection between $s$ and $t$.
	
	\begin{figure}[htbp]
		\centering
		\includegraphics[height=0.32\textheight]{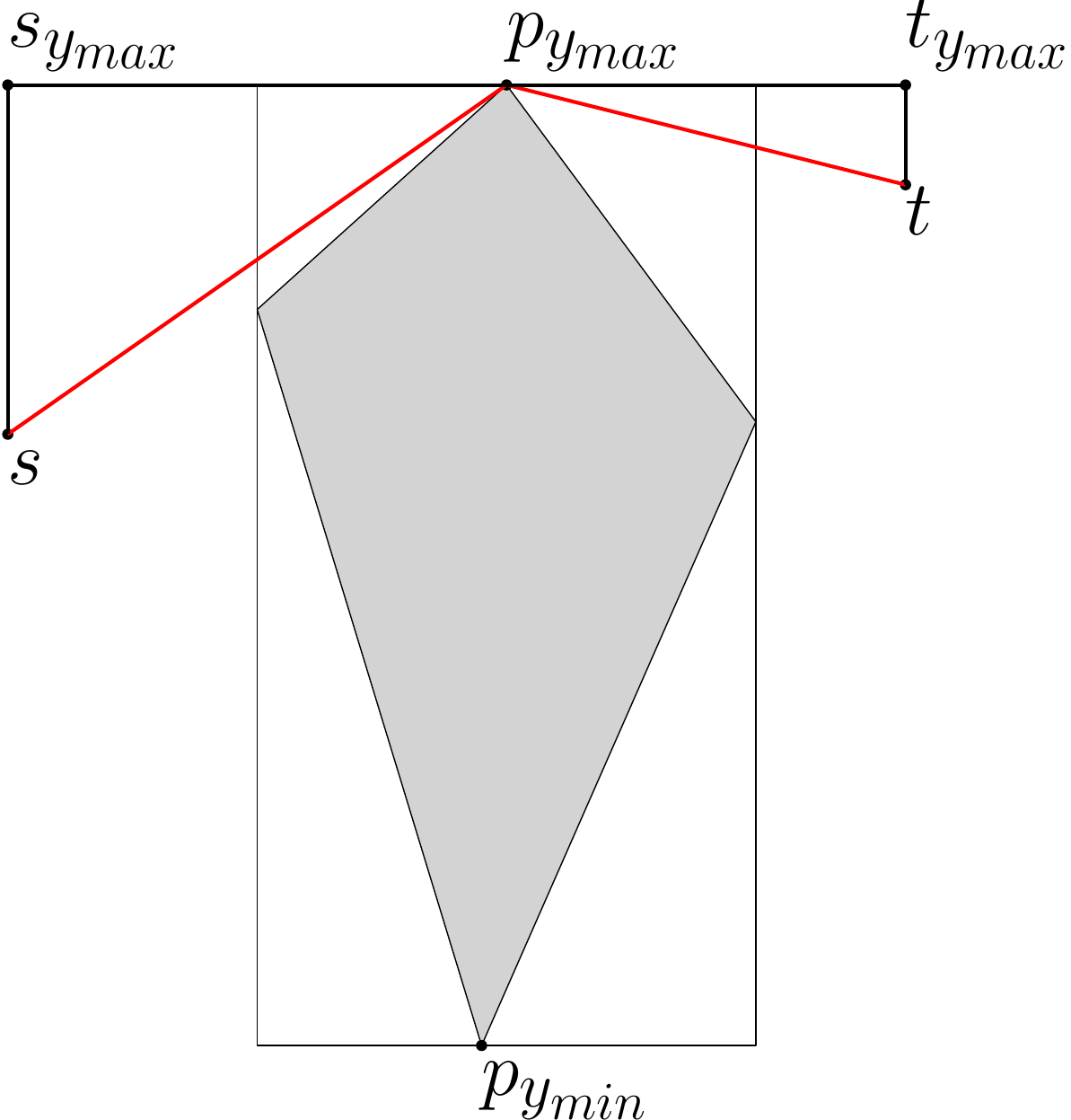}
		\caption{A visualization of Case $2$.}
		\label{fig:singleBBOpp} 
		\hfill
		
	\end{figure}
	
\end{proof}

\noindent
So far, we have analyzed the maximal length of $c$-competitive paths we can achieve for a single bounding box in Bounding Box Visibility Graphs.
It remains to combine these insights with our knowledge about virtual Axes to bound the length of bounding box-paths in \twoDel{}s. \\

\begin{corollary} \label{corollary:singleBB}
	Consider a \twoDel{} $G~=~(V,E)$ which contains a single hole with bounding box $b$. 
	Between any pair of nodes $s$ and $t$ with $s,t \in V$ but $s,t \notin b$, there exists a path $p$ from $s$ to $t$ that contains representatives of bounding boxes such that: 
	\begin{center}
	 $\eucl{p} \leq 5.66 \cdot d_{\mathrm{UDG}}(s,t).$
	\end{center}
	\noindent
	Additionally, there exists an online routing strategy that finds a path $p_{on}$ from $s$ to $t$ such that:
	\begin{center}
	$ \eucl{p_{on}} \leq 7.87 \cdot d_{\mathrm{UDG}}(s,t).$
	\end{center}
\end{corollary}
\noindent
\begin{proof}
	If we reduce the \twoDel{} with $b$ to a Bounding Box Visibility Graph only containing $s,t$ and all nodes of $b$, we know that it contains a path between $s$ and $t$ with length at most $\sqrt{2} \cdot d_{\mathrm{UDG}}(s,t)$.
	If we use the nodes in $G$ that are closest to the nodes of $b$ as representatives for $b$, we can apply virtual Axis routing and obtain that between two adjacent representatives of $b$ a path of length at most $3.996$ times their Euclidean distance exists (\Cref{theorem:virtualAxisPaths}).
	Finally, we can combine both insights and conclude that in \twoDel{}s containing a single hole with bounding box $b$ there is a path between any pair of nodes $s,t \notin b$ of length at most $\sqrt{2} \cdot 3.996 \cdot d_{\mathrm{UDG}}(s,t) \leq 5.66 \cdot d_{\mathrm{UDG}}(s,t)$ which proves the Corollary.
	For the online routing strategy, we can conclude based on \Cref{corollary:onlineRouting} that we can find paths of length $\sqrt{2} \cdot 5.56 \cdot d_{\mathrm{UDG}}(s,t) \leq 7.87 \cdot d_{\mathrm{UDG}}(s,t)$.
	
\end{proof}

\skipSpace
We continue with considering multiple non-intersecting bounding boxes.
\begin{theorem} \label{theorem:nonIntersectingBB}
	Let $G_{BB} = (V_{BB},E)$ be a Bounding Box Visibility Graph that contains multiple non-intersecting bounding boxes and a source- and a target-location $s$ and $t$. 
	There exists a path $p^{BB}_{st}$ between $s$ and $t$ in $G_{BB}$ with:
	\begin{center}
	$\eucl{p^{BB}_{st}} \leq \sqrt{2} \cdot d_{\mathrm{UDG}}(s,t).$
	\end{center}
\end{theorem}
\skipSpace
To prove \Cref{theorem:nonIntersectingBB}, we define a special class of paths in geometric graphs which helps us to construct paths in Bounding Box Visibility Graphs which are $c$-competitive to the shortest path in usual Visibility Graphs.
Therefore, we compare the covered distance in vertical direction as well as the covered distance in horizontal direction of both paths.
\skipSpace
\begin{definition}[Monotone Paths]
	\skipSpace
	A path $p = (p_1,p_2, \dots , p_k)$ in a geometric graph is called \emph{increasing x-monotone} if $x(p_i) \leq x(p_{i+1})$ for all $i \in \{1,\dots , k-1\}$. 
	Analogously, such a path is called \emph{increasing y-monotone} if $y(p_i) \leq y(p_{i+1})$ for all $i \in \{1,\dots , k-1\}$. 
	Similarly, paths are called \emph{decreasing $x$-/$y$-monotone} if $x(p_i) \geq x(p_{i+1})/ y(p_i) \geq y(p_{i+1})$ for all $i \in \set{1, \dots, k-1}$.
	A path is called \emph{$x$-monotone} if it is either increasing or decreasing $x$-monotone.
	Analogously, a path is called \emph{$y$-monotone} if it is either increasing or decreasing $y$-monotone.
\end{definition}
\noindent
For the proof of Theorem \ref{theorem:nonIntersectingBB}, we compare the shortest path $p^{vis}_{st}$ between a pair of nodes $s$ and $t$ in a Visibility Graph $G_{Vis}$  to a path $p^{BB}_{st}$ between $s$ and $t$ in the corresponding Bounding Box Visibility Graph $G_{BB}$. 
Observe that $p^{vis}_{st}$ walks along a sequence of polygons $(p_1, \dots, p_k)$ from $s$ to $t$.
Whenever $p^{vis}_{st}$ walks from a polygon $p_i$ to a polygon $p_{i+1}$, $p^{vis}_{st}$ is $x$- and $y$-monotone for that part, as $p^{vis}_{st}$ follows a direct line segment between $p_i$ and $p_{i+1}$.
The key idea for our proof is to construct a path $p^{BB}_{st}$ in $G_{BB}$ that has the same monotonicity properties as $p^{vis}_{st}$ for every pair of consecutive visited polygons $p_i$ and $p_{i+1}$ of $p^{vis}_{st}$.
Therefore, we introduce a greedy routing strategy for $G_{BB}$ that constructs paths having the same monotonicity properties as $p^{vis}_{st}$. 
Our greedy strategy is called \textit{Greedy Visibility Routing (GreViRo)} and is defined as follows: \\
Let $p^{vis}_{st}$ be a shortest path between two points $s$ and $t$ in a Visibility Graph $G_{Vis}$ that contains polygons with non-intersecting bounding boxes. 
The sequence of polygons visited by $p^{vis}_{st}$ is denoted as $(p_1, \dots, p_k)$ and the direct line segment walked by $p^{vis}_{st}$ from polygon $p_i$ to $p_{i+1}$ is denoted as $p^{vis}_{p_ip_{i+1}}$.
In addition, the intersection points of $p^{vis}_{p_ip_{i+1}}$ with $bb(p_i)$ and $bb(p_{i+1})$ are defined as $i_{p_i}$ and $i_{p_{i+1}}$ respectively.
Further, let $G_{BB}$ be the corresponding Bounding Box Visibility Graph.
Consider a line segment $p^{vis}_{p_ip_{i+1}}$ of $p^{vis}_{st}$.
GreViRo connects two nodes $v_{bb_i}$ and $v_{bb_{i+1}}$ of $G_{BB}$ provided $v_{bb_i}$ and $i_{p_i}$ as well as $i_{p_{i+1}}$ and $v_{bb_{i+1}}$ are visible from each other and the path $(v_{bb_i}, i_{p_i}, i_{p{_i+1}}, v_{bb_{i+1}})$ has the same monotonicity properties as $p^{vis}_{p_ip_{i+1}}$.

\noindent
GreViRo always chooses the node of a bounding box intersecting the face with nodes $v_{bb_i}, i_{p_i}, i_{p{_i+1}}$ and $v_{bb_{i+1}}$ that does not violate the monotonicity properties of $p^{vis}_{p_ip_{i+1}}$ and minimizes the distance to $p^{vis}_{p_ip_{i+1}}$ until $v_{bb_{i+1}}$
is visible.

\begin{figure}[h]
	\centering
	\includegraphics[width=0.65\textwidth]{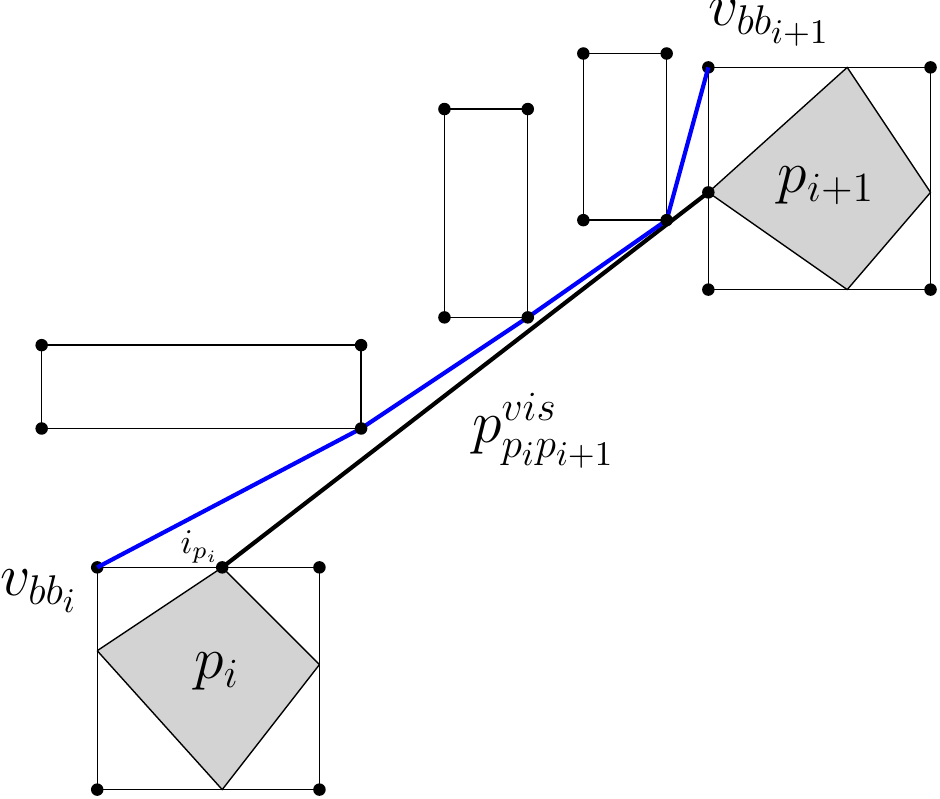}
	\caption{A path construction of GreViRo.}
	\label{fig:greviro} 
\end{figure}
\skipSpace
Figure \ref{fig:greviro} depicts a path construction of GreViRo.
Observe that GreViRo is defined to fulfill the same monotonicity properties as $p^{vis}_{st}$.
We are left with proving the correctness of GreViRo. \\

\begin{lem} \label{lemma:greviro}
	Let $v_{bb_i}$ and $v_{bb_{i+1}}$ be defined as described above.
	GreViRo constructs a path in $G_{BB}$ between $v_{bb_i}$ and $v_{bb_{i+1}}$.
\end{lem}

\begin{proof}
	Without loss of generality, assume $p^{vis}_{p_ip_i+1}$ is increasing $x$- and increasing $y$-monotone.
	Observe that the same proof can be used for all other cases by turning the view $90,180$ and $270$ degrees around.
	Turning the view $90$ degrees around yields an increasing $x$- and decreasing $y$-monotone path, $180$ degrees yields a decreasing $x$- and decreasing $y$-monotone path and $270$ degrees yields a decreasing $x$- and increasing $y$-monotone path. 
	
	\noindent
	We prove by contradiction that GreViRo constructs a path between $v_{bb_i}$ and $v_{bb_{i+1}}$. 
	Assume GreViRo has reached a node $v_{bb_j}$ and cannot proceed further because $v_{bb_{i+1}}$ is not visible and all other visible nodes violate increasing $x$- and $y$-monotonicity.
	Since $v_{bb_{i+1}}$ is not visible from $v_{bb_j}$, there has to be a bounding box $bb_{j+1}$ which is intersected by the line segment $\overline{v_{bb_j}v_{bb_{i+1}}}$.
	Consider a visible node $v_{bb_{j+1}}$ of $bb_{j+1}$.
	Due to our assumption, $x(v_{bb_j}) < x(v_{bb_{j+1}})$ and $y(v_{bb_j}) > y(v_{bb_{j+1}})$.
	Further consider the node $v_{bb_{j-1}}$ which has been visited before proceeding to $v_{bb_j}$. 
	Due to our assumption, GreViRo gets stuck at node $v_{bb_j}$ hence, $x(bb_{j-1}) < x(bb_j)$ and $y(bb_{j-1}) < y(bb_j)$. 
	The crucial observation is that $v_{bb_{j+1}}$ has already been visible from $v_{bb_{j-1}}$.
	Observe that there cannot be a third bounding box intersecting the line segment $\overline{v_{bb_{j-1}}v_{bb_{j+1}}}$ since 
	nodes of this bounding box would have been preferred by GreViRo to $v_{bb_{j+1}}$ as these are closer to $p^{vis}_{p_ip_{i+1}}$.
	Hence, GreViRo would have chosen the node $v_{bb_{j+1}}$ instead of $v_{bb_j}$ which is a contradiction to our assumption. 
	We refer to Figure \ref{fig:xyMonotoneContradiction} for a visualization of the contradiction. \\
\end{proof}
\skipSpace
\begin{figure}[h]
	\centering
	\includegraphics[width=0.75\textwidth]{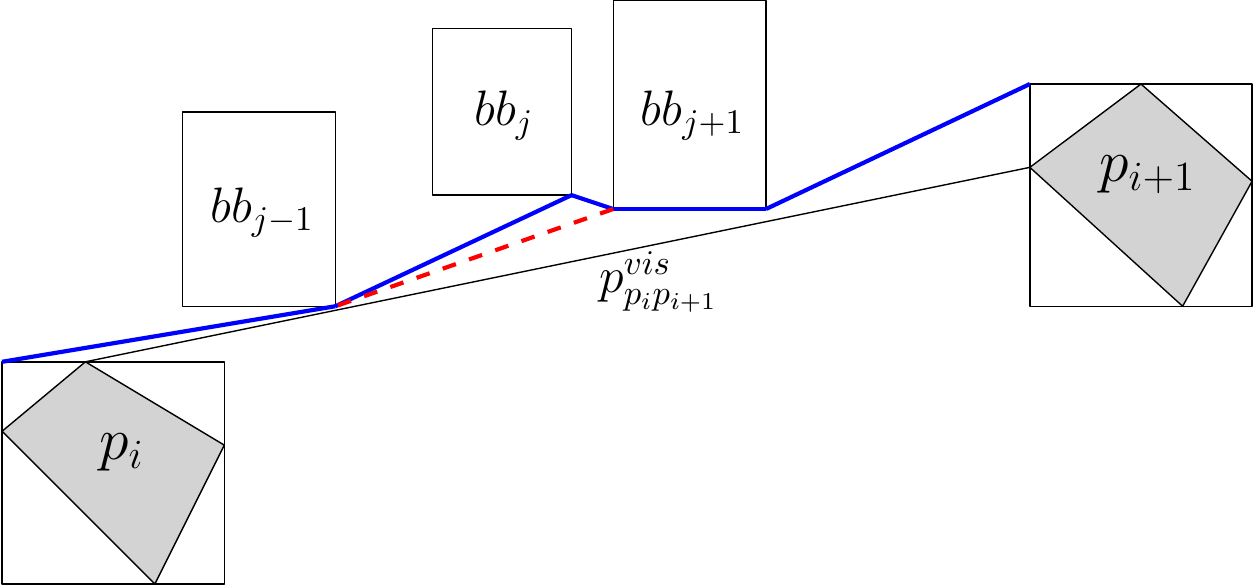}
	\caption[The contradiction described in the proof of Lemma \ref{lemma:greviro}.]{The contradiction described in the proof of Lemma \ref{lemma:greviro}. The contradicting path is marked in blue and thick.
		The reader can see that $bb_{j+1}$ has already been visible from $bb_{j-1}$.}
	\label{fig:xyMonotoneContradiction} 
\end{figure} 

\noindent
GreViRo is an analysis tool that allows us to construct a path in $G_{BB}$ that fulfills the same monotonicity properties as the original path in $G_{Vis}$.
For structuring the proof of Theorem \ref{theorem:nonIntersectingBB}, we split it into two lemmas. 
Initially, we assume that $p^{vis}_{st}$ is $x$- and $y$-monotone (Lemma \ref{lemma:xyMonotone}). 
Using this assumption makes the proof easier and thus helps to understand the overall proof ideas.
Afterward, we drop this assumption in Lemma \ref{lemma:xoryMonotone} and assume that $p^{vis}_{st}$ is either $x$- or $y$-monotone but not both. 
Finally we drop any monotonicity assumption and can prove Theorem \ref{theorem:nonIntersectingBB} with our knowledge of Lemmas \ref{lemma:xyMonotone} and $\ref{lemma:xoryMonotone}$.
\skipSpace
\begin{lem} \label{lemma:xyMonotone}
	If $p^{vis}_{st}$ is $x$- and $y$-monotone, then there exists a path $p^{BB}_{st}$ between $s$ and $t$ in $G_{BB}$ with length at most $\sqrt{2} \cdot d_{\mathrm{UDG}}(s,t)$.
\end{lem}
\newpage
\begin{proof}
	The proof idea is to construct a path $p^{BB}_{st}$ in $G_{BB}$ which is also $x$- and $y$-monotone such that we can conclude that $p^{BB}_{st}$ does neither walk a longer distance in horizontal nor in vertical direction than $p^{vis}_{st}$.
	Without loss of generality, we assume $p^{vis}_{st}$ is both increasing $x$- and $y$-monotone.
	The proofs for the other cases can be obtained by turning the view $90$, $180$ and $270$ degrees around.
	Consider the sequence of hole polygons ($p_1, \dots , p_k$) which is visited by the shortest path $p^{vis}_{st}$.
	Observe that when walking from $p_i$ to $p_{i+1}$, $p^{vis}_{st}$ intersects either the lower edge of $bb(p_{i+1})$ or the left edge of $bb(p_{i+1})$. 
	Note that the edges of $bb(p_{i+1})$ are not part of $G_{Vis}$ but are used here to understand the path construction in $G_{BB}$.
	Due to our monotonicity assumption, $p^{vis}_{st}$ has to intersect either the upper edge of $bb(p_{i+1})$ or the right edge of $bb(p_{i+1})$ when proceeding to $p_{i+2}$. 
	
	\noindent
	The concrete path construction works as follows:
	Assume our path has brought us to $bb(p_i)$ and $p^{vis}_{st}$ is heading to polygon $p_{i+1}$ next.
	There are two cases to consider. Either the left  or the lower edge of $bb(p_{i+1})$ is intersected by $p^{vis}_{st}$.
	Whenever the left edge is intersected, we know due to our monotonicity assumption, that the upper edge has to be intersected afterward. 
	Hence, we can use the top left node of $p_{i+1}$ as next point of our path-construction. 
	The same argumentation can be used if $p_{vis}$ intersects the lower edge of $bb(p_{i+1})$. 
	Due to the monotonicity assumption, we know that the right edge of $bb(p_{i+1})$ has to be intersected when $p^{vis}_{st}$ proceeds to $p_{i+2}$.
	Hence, we can walk to the lower right node of $bb(p_{i+1})$.
	More formally, the next node of our path construction is:
	\begin{itemize}
		\item Case 1: $bb_{t\ell}(p_{i+1})$ if $p^{vis}_{st}$ intersects $\overline{bb_{t\ell}(p_{i+1})bb_{b\ell}(p_{i+1})}$
		\item Case 2: $bb_{br}(p_{i+1})$ if $p^{vis}_{st}$ intersects $\overline{bb_{b\ell}(p_{i+1})bb_{br}(p_{i+1})}$
	\end{itemize}
	
	\noindent
	Note that the path construction only visits the top left or the bottom right node of a bounding box.
	The construction described above yields $4$ different paths from $p_i$ to $p_{i+1}$.
	
	\begin{enumerate}
		\item path from $bb_{t\ell}(p_i)$ to $bb_{t\ell}(p_{i+1})$
		\item path from $bb_{t\ell}(p_i)$ to $bb_{br}(p_{i+1})$
		\item path from $bb_{br}(p_i)$ to $bb_{t\ell}(p_{i+1})$
		\item path from $bb_{br}(p_i)$ to $bb_{br}(p_{i+1})$
	\end{enumerate}
	\noindent
	Unfortunately, it is not guaranteed that the bounding box nodes of $p_i$ and those from $p_{i+1}$ are visible from each other (see~\Cref{fig:visibilityPathMonotone}).
	
	\begin{figure}[h]
		\centering
		\includegraphics[width=0.75\textwidth]{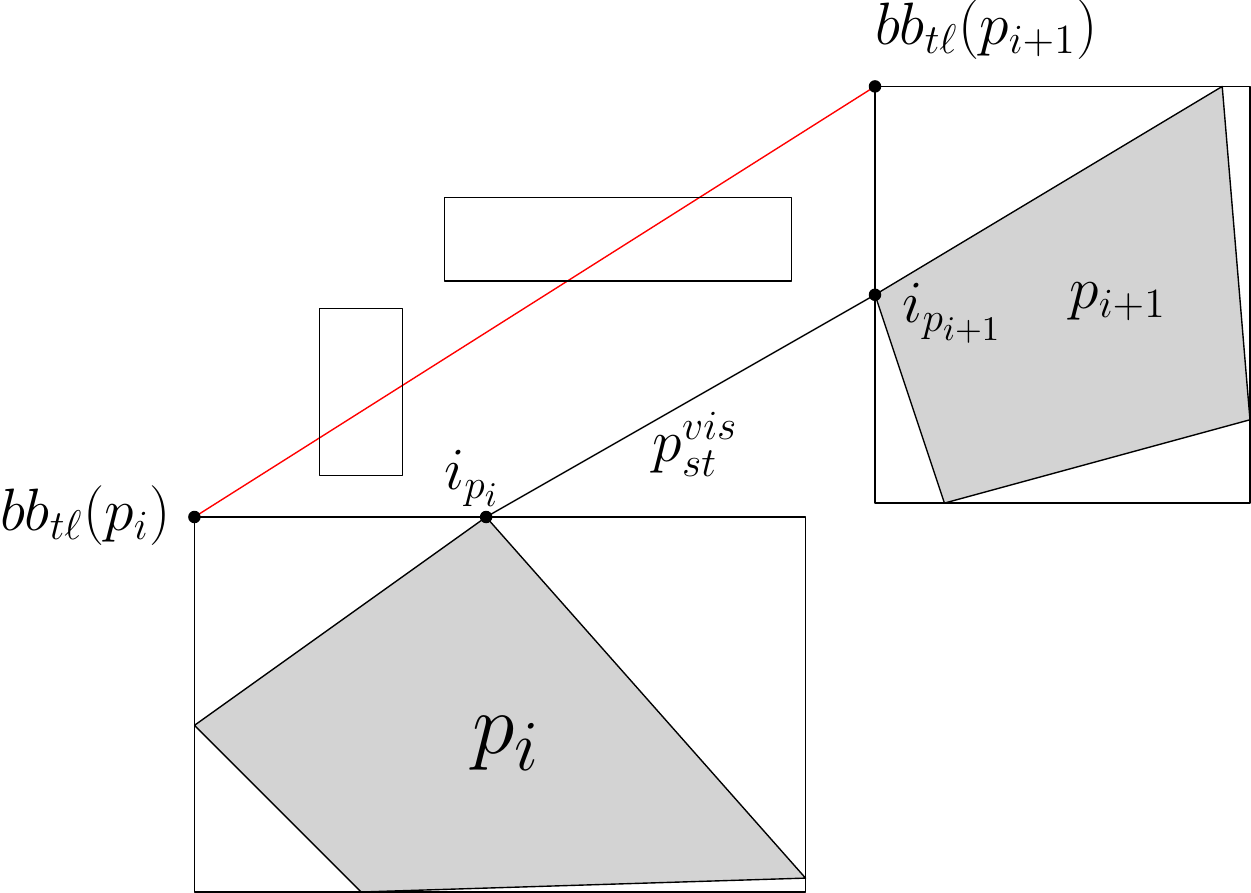}
		\caption[Comparison of $p^{vis}_{st}$ and a potential bounding box path.]{Comparison of $p^{vis}_{st}$ and a potential bounding box path. $p^{vis}_{st}$ uses a direct visibility line whereas $\bbtl{p_i}$ and $\bbtl{p_{i+1}}$ are not visible from each other.}
		\label{fig:visibilityPathMonotone} 
	\end{figure}
	\noindent
	Observe that each case is a valid input for GreViRo. 
	Hence, we can apply GreViRo for every possible case connecting $bb(p_i)$ and $bb(p_{i+1})$ according to our path construction.
	Due to Lemma \ref{lemma:greviro}, $p^{BB}_{st}$ is also increasing $x$- and $y$-monotone.
	If we apply our construction together with GreViRo for every sub-path from $p_i$ to $p_{i+1}$ for all $i \in \set{1, \dots , k-1}$, we obtain an increasing $x$- and $y$-monotone path that passes all polygons which are visited by $p^{vis}_{st}$. 
	The path between $s$ and $p_1$ as well as the path between $p_k$ and $t$ is constructed in the same way.
	Hence, $p^{BB}_{st}$ does neither walk a longer vertical nor horizontal distance than $p^{vis}_{st}$.
	
	\noindent
	Finally, we can compare $p^{BB}_{st}$ to $p^{vis}_{st}$ by using a right triangle with $\eucl{p^{vis}_{st}}$ as length of the hypotenuse $c$ (see Figure \ref{fig:shortestxyMonotonePath}).
	The legs of the right triangle are the parts of $p^{vis}_{st}$ which are walked in horizontal and in vertical direction. 
	Hence, we can conclude that the length of $p^{BB}_{st}$ is at most the sum of both legs.
	Thus, we can bound the length of $p^{BB}_{st}$ by $\sqrt{2} \cdot \eucl{p^{vis}_{st}} \leq \sqrt{2} \cdot d_{\mathrm{UDG}}(s,t)$ (see \Cref{lemma:rightTriangle}) and have finally proven~\Cref{lemma:xyMonotone}.
	\begin{figure}[ht]
		\centering
		\includegraphics[width=0.6\textwidth]{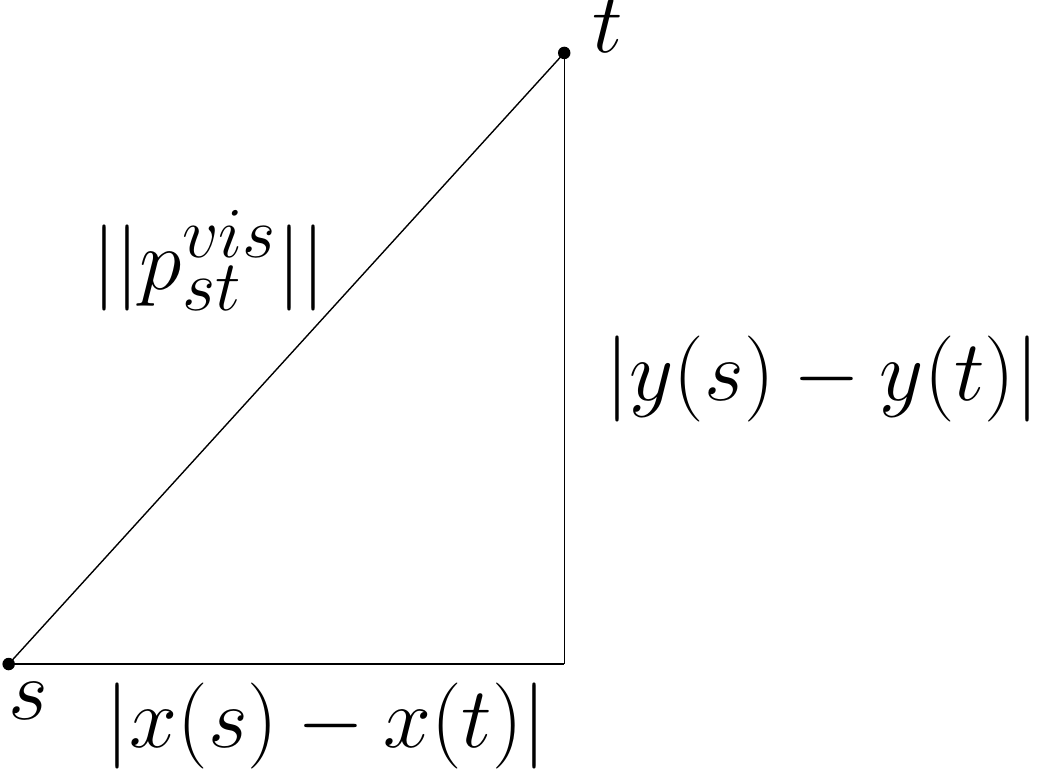}
		\caption[A triangle comparing $\eucl{p^{vis}_{st}}$ and $\eucl{p^{BB}_{st}}$.]{A triangle comparing $\eucl{p^{vis}_{st}}$ and $p^{BB}_{st}$. The length of $p^{BB}_{st}$ is at most the sum of the legs of the depicted right triangle.}
		\label{fig:shortestxyMonotonePath} 
	\end{figure}
\end{proof}
\skipSpace
Unfortunately, shortest visibility paths are not always $x$- and $y$-monotone.
Hence, we have to analyze the path length for more general cases. 
In \Cref{lemma:xoryMonotone}, it is assumed that $p^{vis}_{st}$ is either $x$- or $y$-monotone. 
The proofs of \Cref{lemma:xoryMonotone} and \Cref{theorem:nonIntersectingBB} use similar ideas as the proof of \Cref{lemma:xyMonotone} and are therefore moved to the appendix.\\

\begin{lem} \label{lemma:xoryMonotone}
	If the shortest path $p^{vis}_{st}$ between $s$ and $t$ in $G_{Vis}$ is either $x$- or $y$-monotone, then there exists a path $p^{BB}_{st}$ between $s$ and $t$ in $G_{BB}$ with length at most $\sqrt{2} \cdot d_{\mathrm{UDG}}(s,t)$.
\end{lem}

\begin{proof}
	The overall proof idea is based on the observation that between any visited pair of polygons $p_i$ and $p_{i+1}$ the path $p^{vis}_{st}$ is $x$- and $y$-monotone.
	Our path construction aims to obtain a path that fulfills the same monotonicity criterion as $p^{vis}_{st}$ for any pair of succeeding visited hole polygons.
	Without loss of generality, we assume $p^{vis}_{st}$ is increasing $x$-monotone but not $y$-monotone.
	Turning the view around 90 degrees yields a proof for the other case.
	For the proof of Lemma \ref{lemma:xoryMonotone} we extend the path construction described in the proof of Lemma \ref{lemma:xyMonotone} as follows:
	Assume our path has brought us to the bounding box of hole polygon $p_i$ and we are heading to polygon $p_{i+1}$ next.
	There exist a few more cases which have to be considered for this scenario. 
	One of the new cases (Case $3$) is that $p^{vis}_{st}$ intersects the left edge of $bb(p_{i+1})$ and afterward the lower edge of $bb(p_{i+1})$ since we do not assume $y$-monotonicity anymore.
	In that case, we choose $\bbbl{p_{i+1}}$ as next node.
	Additionally, it can happen that $p^{vis}_{st}$ intersects the upper edge of $bb(p_{i+1})$ and the right edge afterward (Case $4$).
	In that case, $\bbtr{p_{i+1}}$ is chosen.
	
	\noindent
	Furthermore, there is an additional case we have to consider (Case 5).
	As we do not assume $y$-monotonicity anymore, it can happen that $p^{vis}_{st}$ intersects either the upper or the lower edge of $bb(p_{i+1})$ and leaves $bb(p_{i+1})$ afterward through the same edge without intersecting any other edge of $bb(p_{i+1})$.
	We are not allowed to choose a node of $bb(p_{i+1})$ here since we cannot guarantee an $x$- and $y$-monotone path from $p_{i+1}$ to $p_{i+2}$ thereafter (see Figure \ref{fig:xMonotoneNotAllowed} for an example).
	
	\begin{figure}[h]
		\centering
		\includegraphics[width=0.5\textwidth]{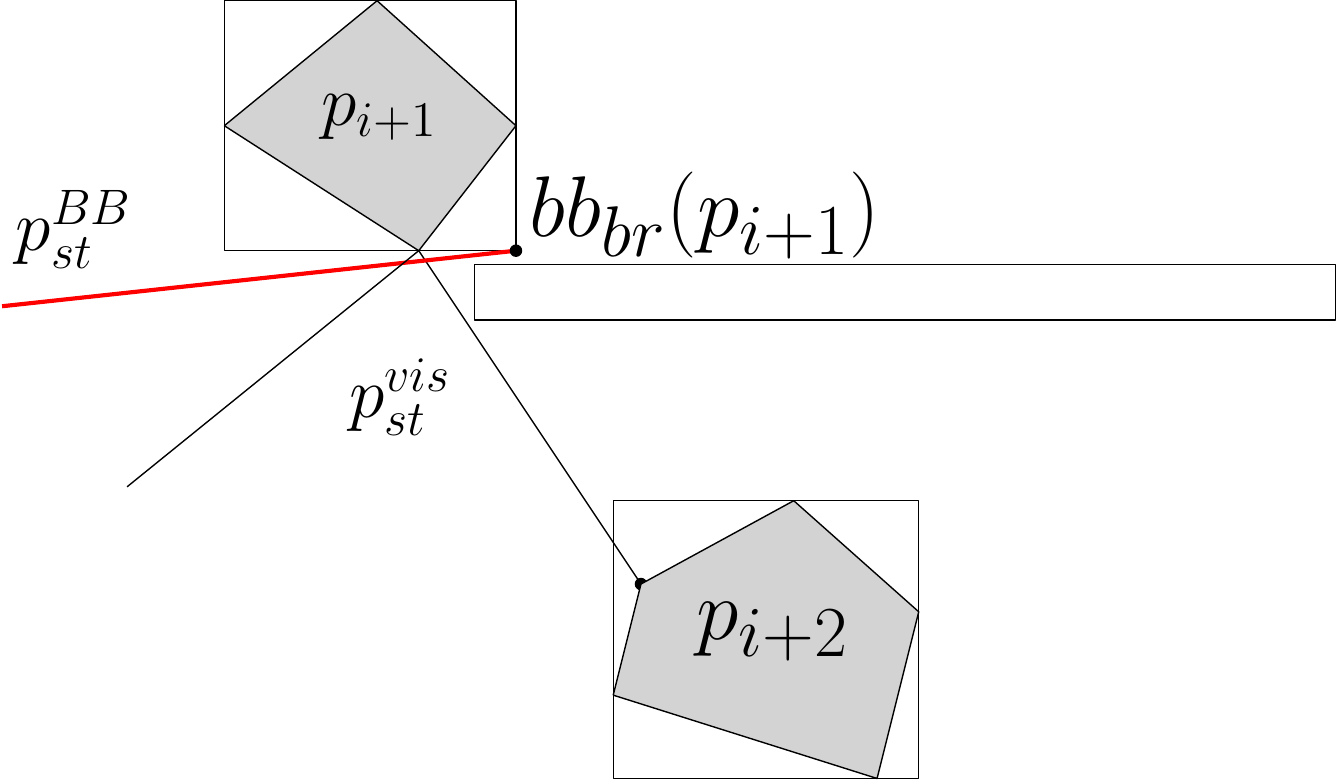}
		\caption[An invalid path construction for Case $5$.]{An invalid path construction for Case $5$. If $p^{BB}_{st}$ chooses $\bbbr{p_{i+1}}$, $p^{BB}_{st}$ has to cover a longer distance in horizontal direction than $p^{vis}_{st}$ when walking from $p_{i+1}$ to $p_{i+2}$.}
		\label{fig:xMonotoneNotAllowed} 
	\end{figure}
	\noindent
	The proposed solution is to avoid using a node of $bb(p_{i+1})$ and instead only walk to a node until the first intersection point of $p^	{vis}_{st}$ with $bb(p_{i+1})$, which is denoted as $i_{p_{i+1},1}$, is visible. 
	That node will be afterward the starting point for the sub-path between $p_{i+1}$ and $p_{i+2}$.
	More formally, we subsume the new cases as follows: 
	Assume we have already covered the sub-path to $p_i$ and $p^{vis}_{st}$ is heading to $p_{i+1}$ next. 
	The next chosen node of our path is: 
	\skipSpace
	\begin{itemize}
		\item Case 3: $bb_{tr}(p_{i+1})$ if $p^{vis}_{st}$ intersects $\overline{bb_{t\ell}(p_{i+1})bb_{tr}(p_{i+1})}$ and afterward \\ $\overline{bb_{tr}(p_{i+1})bb_{br}(p_{i+1})}$
		\item Case 4: $bb_{b\ell}(p_{i+1})$ if $p^{vis}_{st}$ intersects $\overline{bb_{t\ell}(p_{i+1})bb_{b\ell}(p_{i+1})}$ and afterward \\ $\overline{bb_{b\ell}(p_{i+1})bb_{br}(p_{i+1})}$
		\item Case 5: a node of $bb(p_j)$ such that $i_{p_{i+1},1}$ is visible, if $p^{vis}_{st}$ intersects \\
		$\overline{bb_{t\ell}(p_{i+1})bb_{tr}(p_{i+1})}$ or
		$\overline{bb_{b\ell}(p_{i+1})bb_{br}(p_{i+1})}$ twice without intersecting \\
		any other edge of $bb(p_{i+1}).$ 
	\end{itemize}

	\begin{figure}[htbp]
		\begin{minipage}{0.32\textwidth} 
			\includegraphics[width=\textwidth]{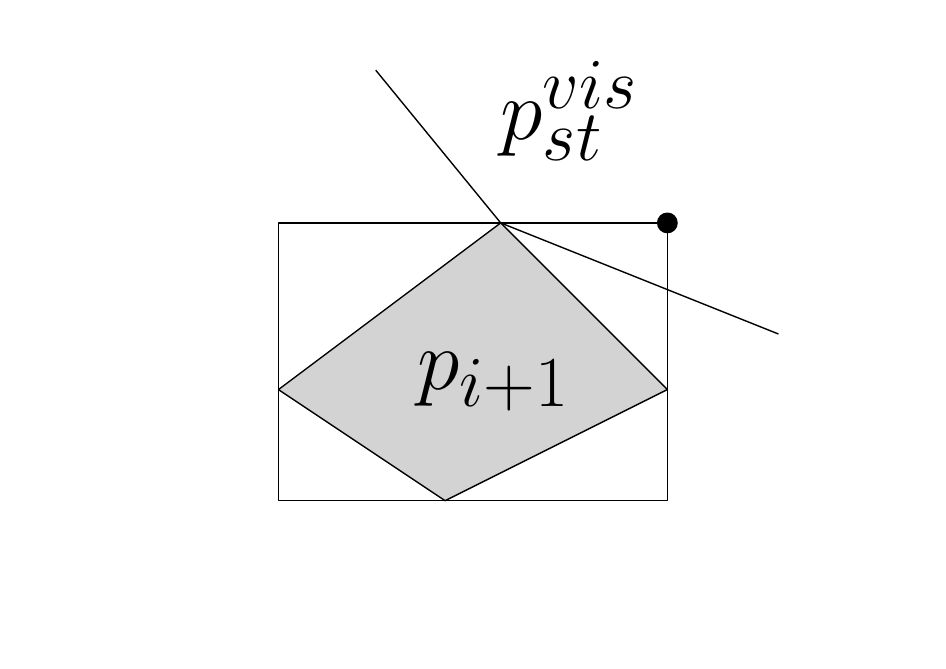}
			\caption{Visualization of Case 3.}
			\label{fig:case3} 
		\end{minipage}
		\hfill
		\begin{minipage}{0.32\textwidth}
			\includegraphics[width=\textwidth]{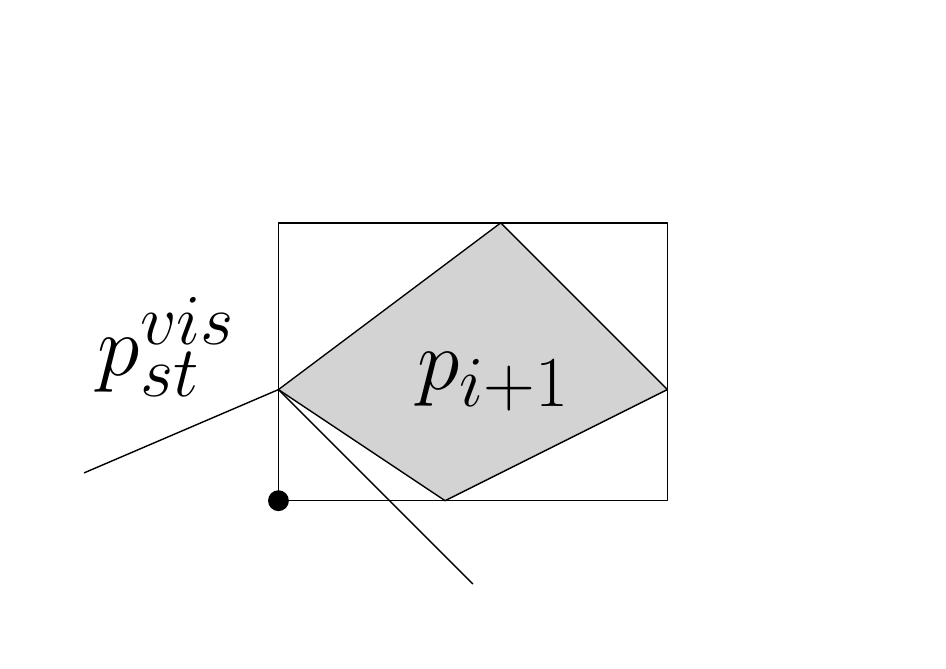}
			\caption{Visualization of Case 4.}
			\label{fig:case4} 
		\end{minipage}
		\hfill
		\begin{minipage}{0.32\textwidth}
			\includegraphics[width=\textwidth]{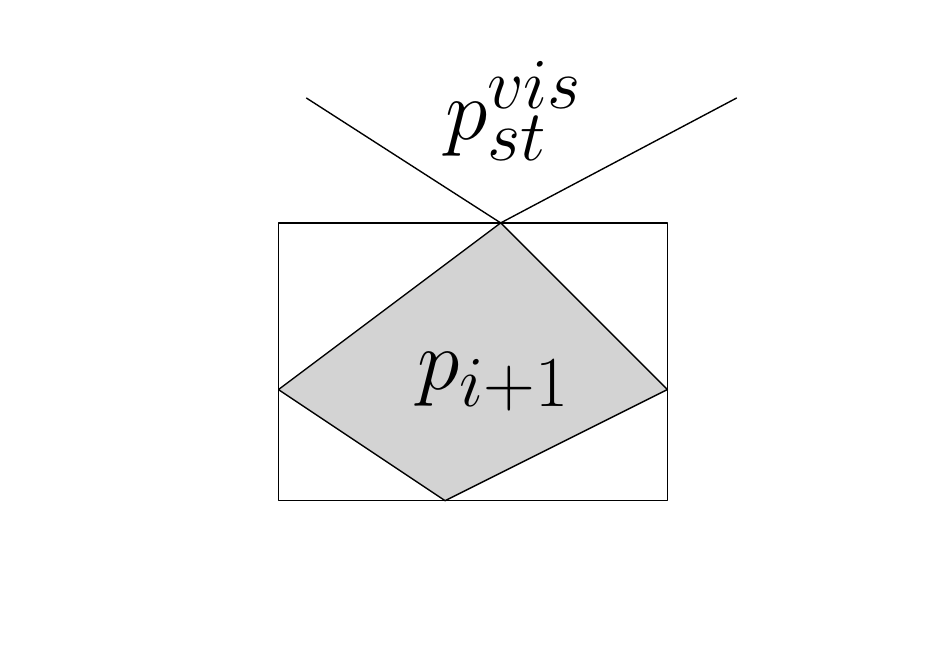}
			\caption{Visualization of Case 5.}
			\label{fig:case5} 
		\end{minipage}
	\end{figure}
	\noindent
	Figures \ref{fig:case3}, \ref{fig:case4} and \ref{fig:case5} provide visualizations of each individual case. \\
	\noindent
	Case $3$ can be handled with similar arguments we have used for the path construction in the proof of Lemma \ref{lemma:xyMonotone}.
	Case $4$ needs a few additional arguments, as $p^{vis}_{st}$ does not fulfill  $y$-monotonicity.
	Without loss of generality, assume that we are currently at $\bbtl{p_i}$ and are heading to $\bbbl{p_{i+1}}$ next.
	Note that this implies that $p^{vis}_{st}$ is increasing $y$-monotone from $p_{i}$ to $p_{i+1}$.
	We construct a path with GreViRo for two sub-paths between $\bbtl{p_i}$ and $\bbbl{p_{i+1}}$.
	The first sub-path is increasing $x$- and increasing $y$-monotone.
	The second path stays increasing $x$-monotone but changes to decreasing $y$-monotonicity. 
	We apply GreViRo along $p^{vis}_{st}$ until  $\bbbl{p_{i+1}}$ is reached or another node $v$ would be chosen with $y(v) > y(i_{p_{i+1}})$.
	In the second case the strategy is changed and GreViRo is applied to the line segment $\overline{i_{p_{i+1}}\bbbl{p_{i+1}}}$. 
	Observe that this can only happen at a node which is visible from $i_{p_{i+1}}$.
	Hence, we have a valid input for GreViRo for the second sub-path.
	As we have already proven in Lemma \ref{lemma:greviro}, the first part of the path is increasing $x$- and increasing $y$-monotone and the second path is increasing $x$- and decreasing $y$-monotone. 
	Hence, $p^{BB}_{st}$ fulfills the same monotonicity properties as $p^{vis}_{st}$ between $p_i$ and $p_{i+1}$.
	The proof for other starting nodes than $\bbtl{p_i}$ is a simple adaptation of the given proof and omitted here. 
	
	\noindent
	It remains to prove a path existence for Case $5$.
	For the proof, we initially assume that from $p_i$ to $p_{i+1}$ Case $5$ occurs for the first time and hence we are either at $\bbtl{p_i}, \bbtr{p_i}, \bbbl{p_i}$ or $\bbbr{p_i}$.
	Without loss of generality, we assume that we are at $\bbtl{p_{i}}$. 
	Additionally we assume without loss of generality that $p^{vis}_{st}$ intersects the lower edge of $bb_{i+1}$.
	When applying GreViRo, we eventually reach a node of $bb(p_j)$ such that $i_{p_{i+1,1}}$ is visible.
	Consequently, we achieve an increasing $x$- and increasing $y$-monotone sub-path to a node of $bb(p_j)$.
	Nonetheless, our proof is not finished yet. 
	The remaining problem is the part between $bb(p_j)$ until we can walk along $p^{vis}_{st}$ again (see Figure \ref{fig:xoryMonotoneSubpath} for a visualization of the problem). 
	
	\begin{figure}[h]
		\centering
		\includegraphics[width=0.5\textwidth]{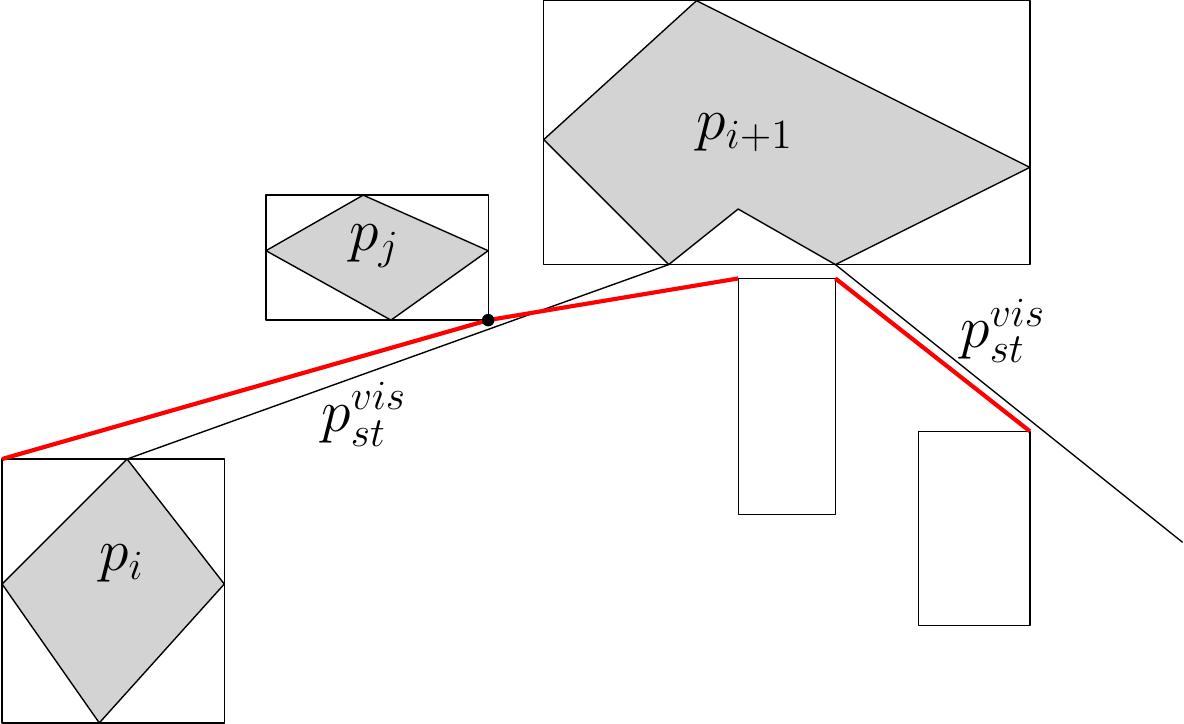}
		\caption[GreViRo example for Case $5$.]{GreViRo example for Case $5$. After reaching the node of $p_j$, GreViRo has to be applied until $i_{p_{i+1},2}$ is visible.} 
		\label{fig:xoryMonotoneSubpath}
	\end{figure}
	\noindent
	The solution is to introduce another sub-path. 
	This construction works as follows. 
	We start at $bb(p_j)$ and want to reach a node of $bb(p_{j'})$ such that $i_{p_{i+1},2}$ is visible.
	After reaching that point we can walk along $p^{vis}_{st}$ with our already known greedy strategy.
	Observe that this is a valid input for GreViRo.
	By applying GreViRo, we obtain an increasing $x$- and increasing $y$-monotone sub-path between $bb(p_j)$ and $bb(p_{j'})$.
	Finally, we can plug all these sub-paths together. 
	The argumentation for all other cases is similar to the discussed case and are therefore omitted here.
	Hence, we obtain:
	Each sub-path of $p^{BB}_{st}$ does not cover a longer distance in horizontal or vertical direction than $p^{vis}_{st}$. 
	Therefore, we are able to use the right triangle of Lemma \ref{lemma:xyMonotone} again (see Figure \ref{fig:shortestxyMonotonePath}).
	Consequently: 
	\begin{center}
	$\eucl{p^{BB}_{st}} \leq \sqrt{2} \cdot \eucl{p^{vis}_{st}} \leq \sqrt{2} \cdot d_{\mathrm{UDG}}(s,t).$
	\end{center}
\end{proof}	
\noindent
With help of \Cref{lemma:xyMonotone,lemma:xoryMonotone}, we can prove the correctness of Theorem \ref{theorem:nonIntersectingBB}.

\begin{proof}[Proof of \Cref{theorem:nonIntersectingBB}]
	By dropping the last monotonicity assumption, we obtain a few more cases again.
	Assume our path has brought us to $bb(p_i)$ and we are heading to $bb(p_{i+1})$ next.
	The first new case (Case $6$) is that the right edge of bounding box $bb(p_{i+1})$ is intersected and afterward the lower edge.
	In that case we walk to $\bbbr{p_{i+1}}$.
	The second new case (Case $7$) is that the right edge of $bb(p_{i+1})$ is intersected first and afterward the upper edge.
	Hence, we can walk to $\bbtr{p_{i+1}}$.
	Case $8$ is similar to Case $5$ of Lemma \ref{lemma:xoryMonotone}. 
	In that case, either the left or the right edge of $bb(p_{i+1})$ is intersected twice without intersecting any other edge of $bb(p_{i+1})$.
	Therefore, we use the same technique we have already used for the proof of Lemma \ref{lemma:xoryMonotone}.
	We walk as close as possible to $bb(p_{i+1})$ such that $i_{p_{i+1},1}$ is visible.
	More formally, the path construction chooses the following node of $bb(p_{i+1})$:
	\begin{itemize}
		\item Case 6: $\bbbr{p_{i+1}}$ if $p^{vis}_{st}$ intersects $\overline{\bbtr{p_{i+1}}\bbbr{p_{i+1}}}$ first and afterward $\overline{\bbbl{p_{i+1}}\bbbr{p_{i+1}}}$
		\item Case 7: $\bbtr{p_{i+1}}$ if $p^{vis}_{st}$ intersects $\overline{\bbtr{p_{i+1}}\bbbr{p_{i+1}}}$ \\first and afterward $\overline{\bbtr{p_{i+1}}\bbbr{p_{i+1}}}$
		\item Case 8: A node of $bb(p_j)$ such that $i_{p_{i+1},1}$ is visible, if $p^{vis}_{st}$ intersects \\
		$\overline{bb_{tr}(p_{i+1})bb_{br}(p_{i+1})}$ or
		$\overline{bb_{t\ell}(p_{i+1})bb_{b\ell}(p_{i+1})}$ twice without intersecting any other edge of $bb(p_{i+1})$ 
	\end{itemize}
	\noindent
	Figures \ref{fig:case6}, \ref{fig:case7} and \ref{fig:case8} visualize each case.
	\begin{figure}[htbp]
		\begin{minipage}{0.32\textwidth} 
			\includegraphics[width=\textwidth]{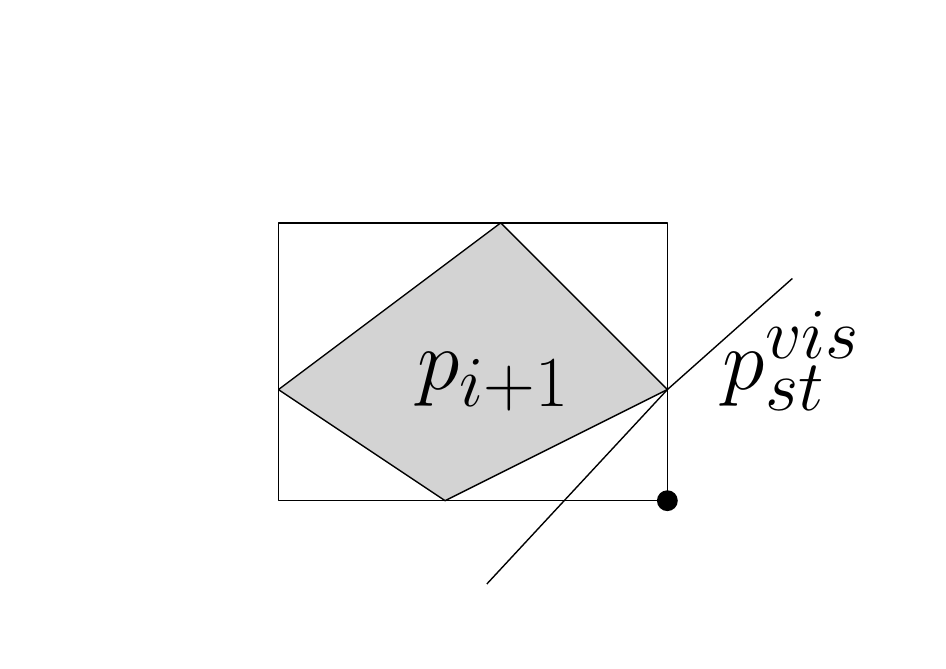}
			\caption{Visualization of Case 6.}
			\label{fig:case6} 
		\end{minipage}
		\hfill
		\begin{minipage}{0.32\textwidth}
			\includegraphics[width=\textwidth]{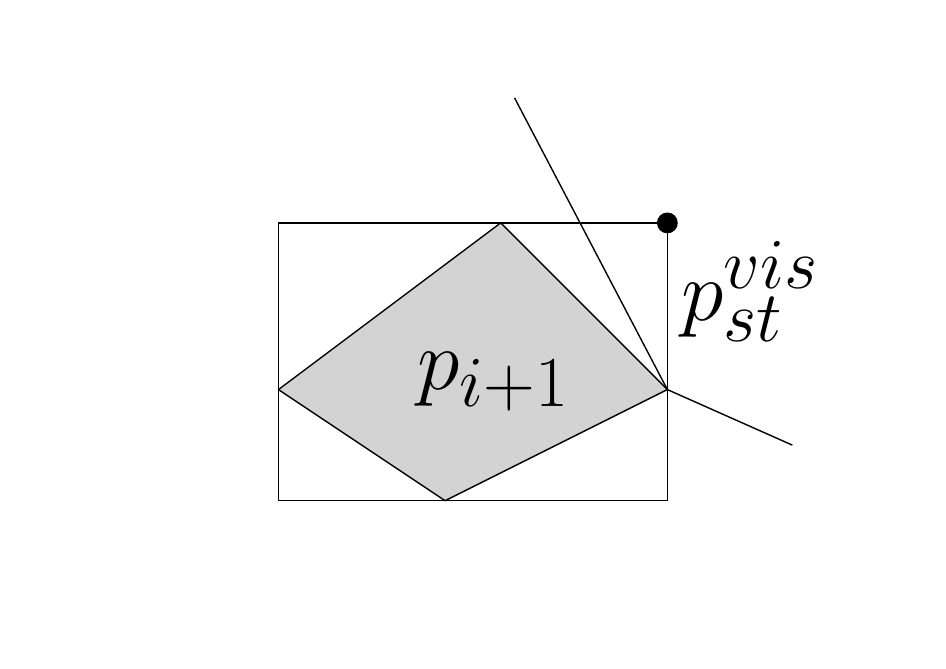}
			\caption{Visualization of Case 7.}
			\label{fig:case7} 
		\end{minipage}
		\hfill
		\begin{minipage}{0.32\textwidth}
			\includegraphics[width=\textwidth]{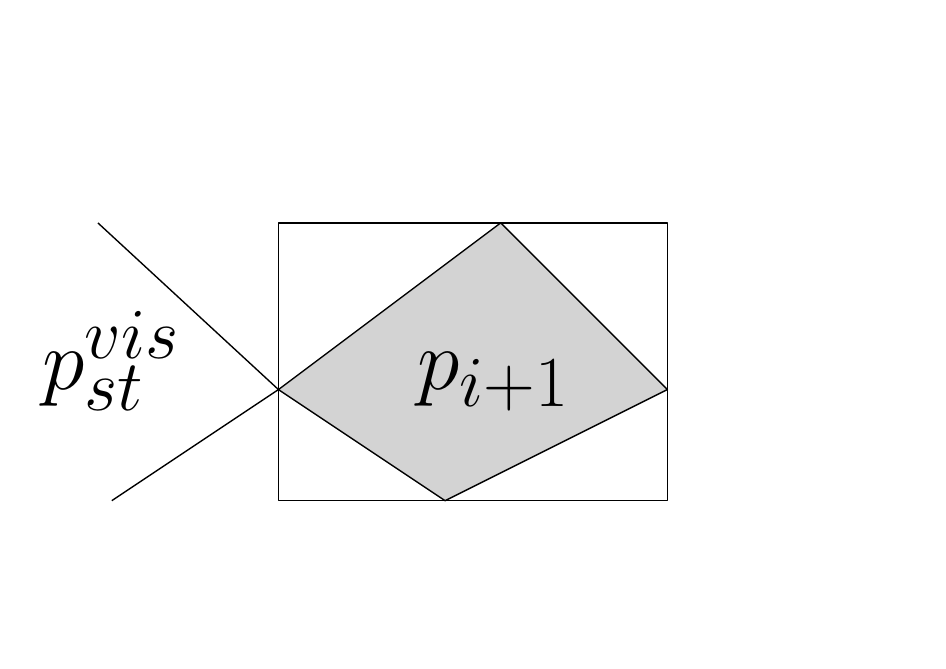}
			\caption{Visualization of Case 8.}
			\label{fig:case8} 
		\end{minipage}
	\end{figure}
	\noindent
	Observe that all of the introduced cases are either point or axis reflections of the path constructions we have already proven for Lemmas \ref{lemma:xyMonotone} and $\ref{lemma:xoryMonotone}$.
	Hence, we can construct a path that fulfills the same monotonicity properties as $p^{vis}_{st}$ between each pair of succeeding hole polygons $p_i$ and $p_{i+1}$ that are visited by $p^{vis}_{st}$.
	Thus, there exists a path $p^{BB}_{st}$ in $G_{BB}$ with:
	\begin{center}
	$\eucl{p^{BB}_{st}} \leq \sqrt{2} \cdot d_{\mathrm{UDG}}(s,t).$
	\end{center}
	
\end{proof}
If we consider $m$ non-intersecting bounding boxes in a \twoDel{} and assume that each node of a bounding box is represented by the closest node of the \twoDel{}, we can apply virtual Axis Routing between any two adjacent nodes on the path $p^{BB}_{st}$.
Hence, we obtain a path in with length at most $\sqrt{2} \cdot 3.996$ $d_{\mathrm{UDG}}(s,t) = 5.66 \cdot d_{\mathrm{UDG}}(s,t)$.
For online routing, we obtain a path of length at most $\sqrt{2} \cdot 5.56$ $d_{\mathrm{UDG}}(s,t) = 7.87 \cdot d_{\mathrm{UDG}}(s,t)$.
This property is stated by Corollary \ref{corollary:multipleBB}.
\skipSpace
\begin{corollary} \label{corollary:multipleBB}
	Consider a \twoDel{} $G = (V,E)$ which contains up to $m$ holes.
	Between any pair of nodes $s$ and $t$ with $s,t \in V$ that do not lie in any bounding box, there exists a path $p$ from $s$ to $t$ such that: 
	\begin{center}
	$ \eucl{p} \leq 5.66 \cdot d_{\mathrm{UDG}}(s,t).$
	\end{center}
	Additionally, applying MixedChordArc along every edge on the shortest path between $s$ and $t$ in the Bounding Box Visibility Graph of the ad hoc network leads to a path $p_{on}$ with 
		\begin{center}
	$\eucl{p_{on}} \leq 7.87 \cdot d_{\mathrm{UDG}}(s,t).$
		\end{center}
\end{corollary}
\noindent
After our extensive studies of non-intersecting bounding boxes we switch our focus to intersecting bounding boxes. 

\subsection{Competitive Paths via intersecting Bounding Boxes} \label{section:intersectingBB}
In this section we drop the last restriction and consider intersecting bounding boxes of holes.
This leads to entirely new challenges. 
What if the shortest path between nodes $s$ and $t$ leads through an area in which two or more bounding boxes intersect?
We are faced with two major problems in this setting.
The first problem which could occur is that nodes of bounding boxes could lie in holes.
See \Cref{fig:multipleBB} for a visualization of the problem.
The reader can easily see that it can happen that all nodes of bounding boxes which lie inside an area of intersecting bounding boxes could lie in holes.
In such cases, we cannot gain any information out of these nodes.
Therefore, we drop all of these nodes and only keep nodes located on the outer boundary of all intersecting bounding boxes.

\begin{figure}[htbp]
	\centering
	\includegraphics[width=0.45\textwidth]{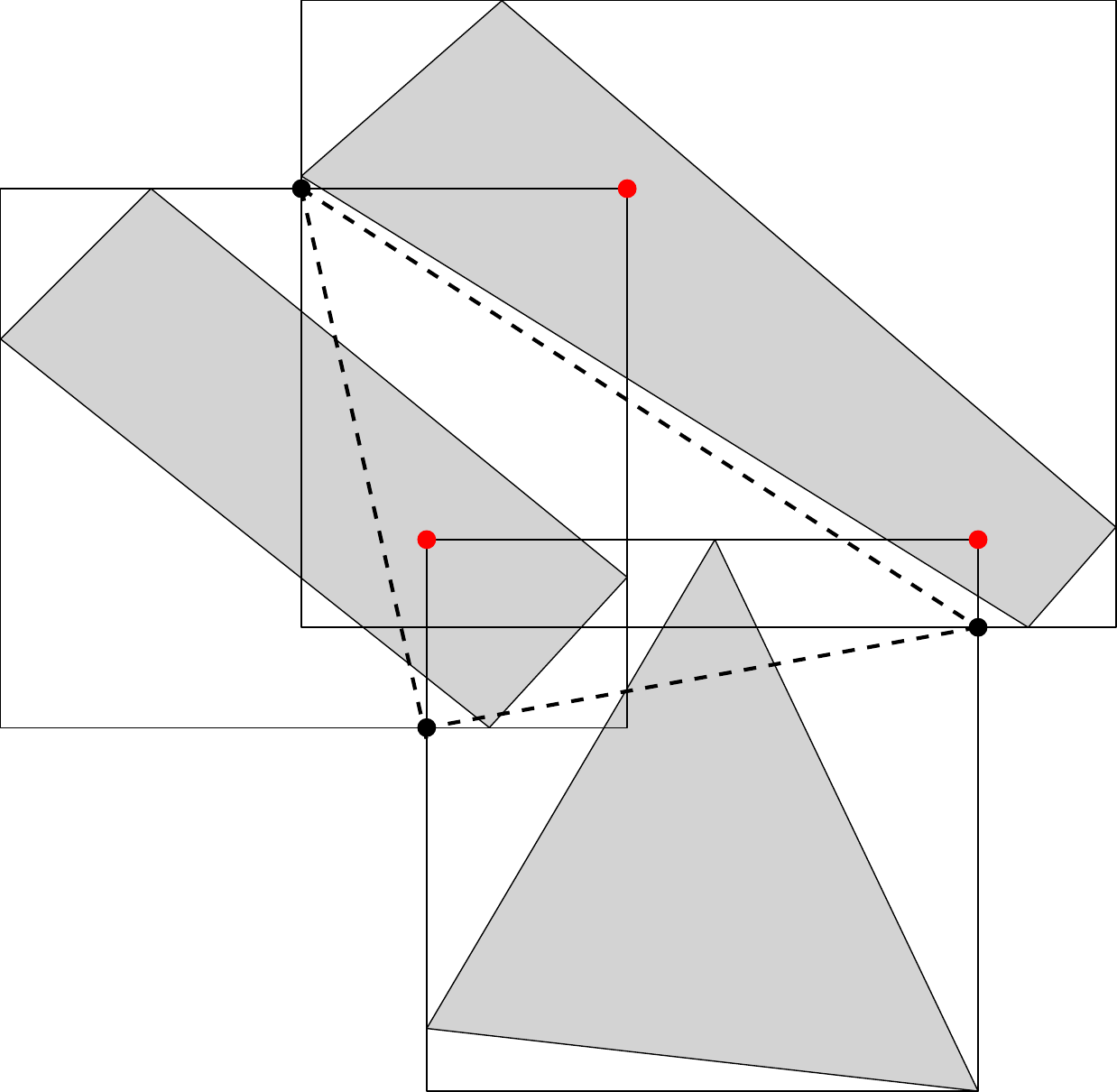}
	\caption{Multiple intersecting bounding boxes with nodes lying in holes.}
	\label{fig:multipleBB} 
	\hfill
	
\end{figure}
\noindent	
It is easy to see that routing is not complicated as long as the shortest path avoids such situations.
In these cases we can apply the result of Theorem \ref{theorem:nonIntersectingBB} and also obtain a $\sqrt{2}$-competitive path.
Hence, we only have to focus on cases in which the shortest path leads through an area of two or more intersecting bounding boxes.

\noindent
This leads over to the second major problem we have to consider in this scenario.
Unfortunately, a shortest path that leads through an area of two or more intersecting bounding boxes can be arbitrarily complex as Figure \ref{fig:halloween} depicts.
In such kinds of situations, we cannot find $c$-competitive paths by only using information obtained by bounding boxes.
Consequently, we have to enrich the information contained in the Bounding Box Visibility Graph.

\begin{figure}[h]
	\centering
	\includegraphics[width=0.75\textwidth]{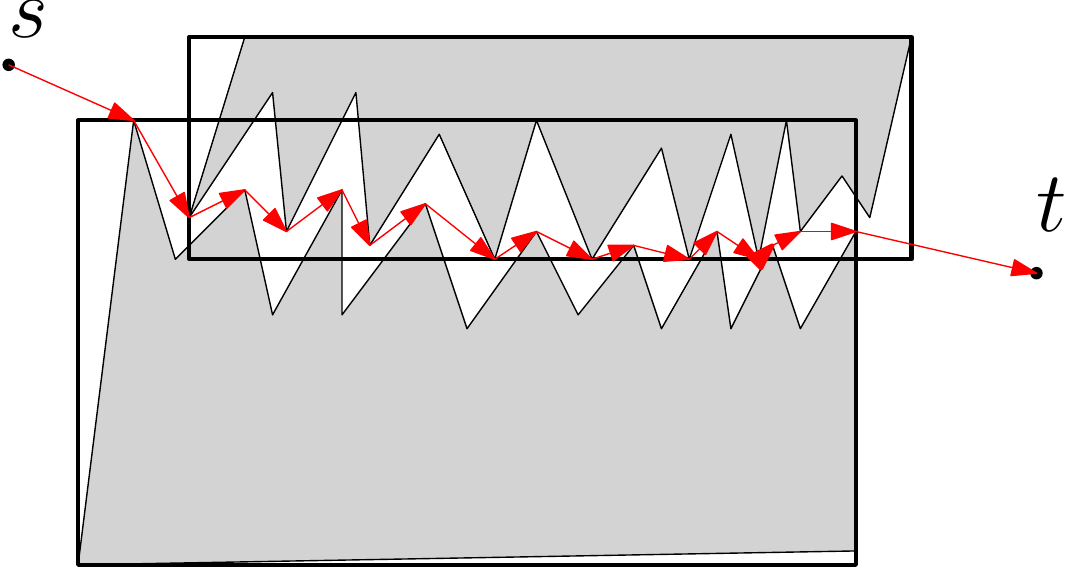}
	\caption[An area of intersecting bounding boxes.]{An area of intersecting bounding boxes. A shortest path which leads through the intersecting area of two bounding boxes can be arbitrarily complex.}
	\label{fig:halloween} 
\end{figure}

\noindent
Consider the outer boundary of an area where multiple bounding boxes intersect.
On the outer boundary, we can find nodes of bounding boxes and additionally, there are intersection points of bounding boxes.
We call these points \textit{outer intersection points}.
The crucial observation is that whenever a shortest path leads through an area in which at least two bounding boxes intersect, such an outer intersection point has to be passed in vertical and also in horizontal direction.
Due to this observation, we modify Bounding Box Visibility Graphs as follows:
Whenever two or more bounding boxes intersect, we keep those nodes which lie on the outer boundary of that area and drop all nodes which lie inside.
Further, we insert all outer intersection points into the node set.
Outer intersection points of the same area are connected in a clique.
The weight of an edge that connects outer intersection points $o_1$ and $o_2$ is the length of the shortest path in the corresponding Visibility Graph connecting $o_1$ and $o_2$ inside of the intersection area.
The described construction is called \textit{modified Bounding Box Visibility Graph}. 
See Figure \ref{fig:multipleBBClique} for an exemplary modified Bounding Box Visibility Graph.
The rest of the section deals with proving Theorems \ref{theorem:intersectingBBLowerBound} and \ref{theorem:modifiedBBVisibility} which state upper and lower bounds for $c$-competitive paths in modified Bounding Box Visibility Graphs.

\begin{figure}[h]
	\centering
	\includegraphics[width=0.55\textwidth]{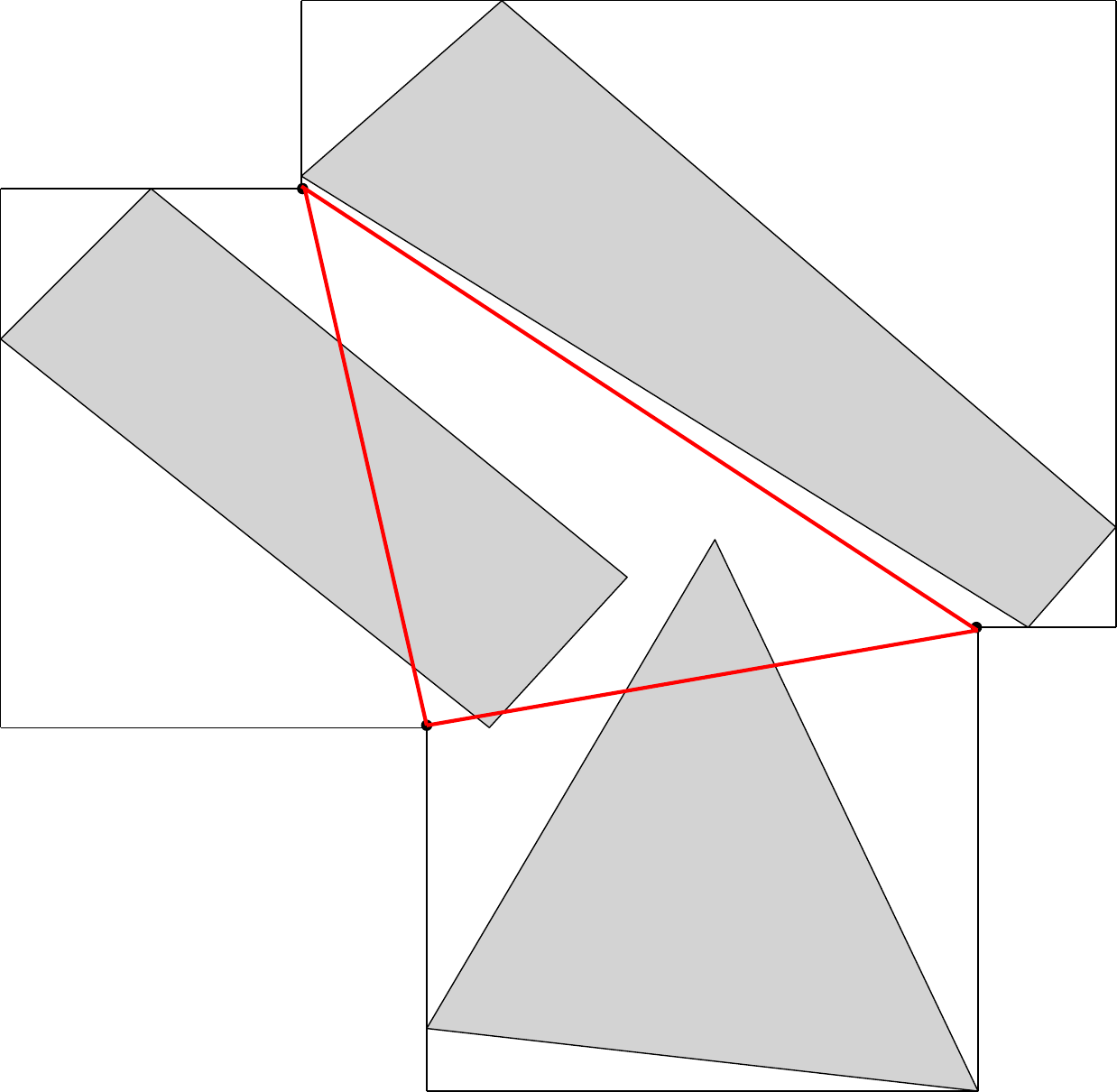}
	\caption[The clique of outer intersection points.]{The clique of outer intersection points (marked in red).}
	\label{fig:multipleBBClique} 
\end{figure}

\noindent
We start with giving a lower bound on the competitive ratio which can be obtained by modified Bounding Box Visibility Graphs.  \\
\begin{theorem} \label{theorem:intersectingBBLowerBound}
	There exist modified Bounding Box Visibility Graphs which contain only paths between a source-node $s$ and a target-node $t$ with length at least $2.8 \cdot d_{\mathrm{UDG}}(s,t)$.
\end{theorem}

\begin{proof}
	We prove Theorem \ref{theorem:intersectingBBLowerBound} by giving a construction in which the lower bound is achieved.
	The considered scenario is visualized in Figure \ref{fig:lowerBoundIntersecting}.
	The shortest path between $s$ and $t$ leads through a bounding box that has a lower edge of length $x$.
	Hence, we can place $s$ and $t$ such that the distance $\eucl{st}$ is arbitrarily close to $\sqrt{2}\cdot x$ by guaranteeing $|x(s)-x(t)| \approx x$ and $|y(s)- y(t)| \approx x$.
	The scenario contains a second bounding box which intersects only the lower side of the first bounding box.
	The bounding box and the hole polygon coincide for that box.
	The intersection of both boxes yields two intersection points $o_1$ and $o_2$. 
	Assume that the first bounding box is high enough for the shortest path between $s$ and $t$ to choose $o_1$ and $o_2$ as intermediate points.
	Since edges of bounding boxes have length at least one, we can place the second bounding box such that the length of the line segment $\overline{so_1}$ is arbitrarily close to $x-1$.
	Further, the line segment $\overline{o_1o_2}$ can have a length arbitrarily close to one and the height of the bounding box in the intersecting area can have a length arbitrarily close to $x-1$.
	Hence, we obtain a path length between $s$ and $t$ of $(x-1)+(x-1)+1+(x-1)+(x-1) = 4x-3$.
	The shortest geometric path is the direct line segment $\overline{st}$ with length arbitrarily close to $\sqrt{2}\cdot x$ by ensuring $\overline{st} = \sqrt{x^2 +x^2} = \sqrt{2x^2} = \sqrt{2} \cdot x$.
	Thus, we can compute the ratio of both paths and its limits:
	\begin{center}
	$\lim\limits_{x \rightarrow \infty}{\frac{4x-3}{\sqrt{2}\cdot x}} 
	\phantom{test}\overset{L'Hospital}{=}\lim\limits_{x \rightarrow \infty}{\frac{4}{\sqrt{2}}}\leq 2.83$
	\end{center}
	\noindent
	The calculation follows from L'Hospital's rule which is applicable because $4x-3$ and $\sqrt{2}\cdot x$ both diverge.
	Hence, we have proven a lower bound for the competitive ratio of paths which can be computed with modified Bounding Box Visibility Graphs. \\
	\skipSpace 
\end{proof}
\begin{figure}[h]
	\centering
	\includegraphics[width=0.5\textwidth]{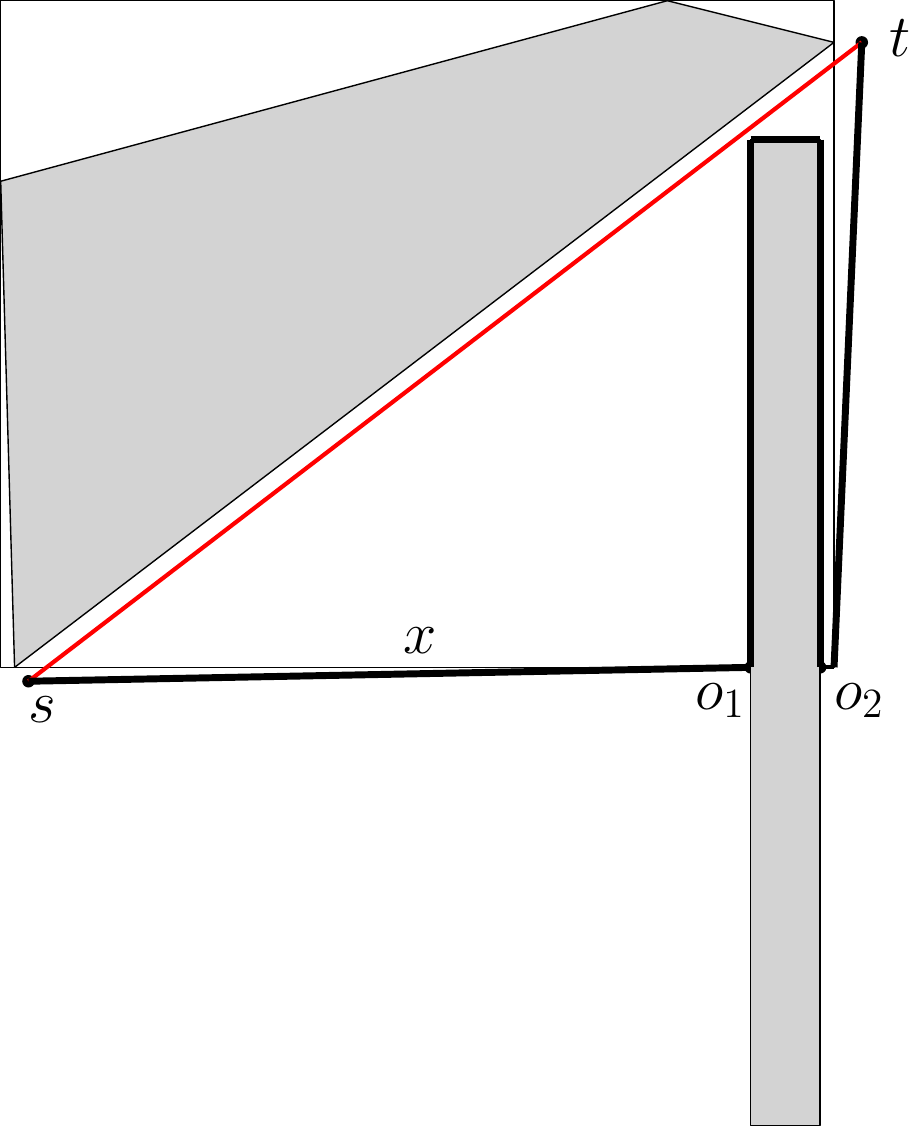}
	\caption[Visualization of the lower bound.]{Visualization of the lower bound. The shortest possible geometric connection between $s$ and $t$ is the direct line segment $\overline{st}$ drawn in red.
		The path taken in the modified Bounding Box Visibility Graph is marked in black (thick).}
	\label{fig:lowerBoundIntersecting}
\end{figure}
\noindent
After proving a lower bound, we continue with proving an upper bound. \\

\begin{theorem} \label{theorem:modifiedBBVisibility}
	In modified Bounding Box Visibility Graphs, there exists a path $p^{BB}_{st}$ between two nodes $s$ and $t$, provided $s$ and $t$ lie outside of every bounding box, such that:
	\begin{center} 
	$\eucl{p^{BB}_{st}} \leq 4.42 \cdot d_{\mathrm{UDG}}(s,t).$
	\end{center}
\end{theorem}
\skipSpace
\begin{proof}
	Consider the Visibility Graph $G_{Vis}$ and the corresponding modified Bounding Box Visibility Graph $G_{BB}$.
	Let $s$ be a start- and $t$ a target-location such that the shortest path $p^{vis}_{st}$ between $s$ and $t$ in $G_{Vis}$ leads through an area of polygons whose bounding boxes intersect.
	The proof is based on two main observations.
	The first observation is that $p^{vis}_{st}$ passes at least the $x$- and $y$- coordinates of two outer intersection points. Observe that this does not imply that $p^{vis}_{st}$ visits these points.
	Since we assume that $p^{vis}_{st}$ leads through an area of intersecting bounding boxes, we know that these outer intersection points have to be passed when $p^{vis}_{st}$ enters and leaves that area.
	We denote the first outer intersection point with $o_1$ and the second one with $o_2$. In case there are more than two outer intersection points, $o_2$ denotes the last one.
	Assume we would add $o_1$ and $o_2$ to $G_{Vis}$ and consider the sub-paths $p^{vis}_{so_1}$, $p^{vis}_{o_1o_2}$ and $p^{vis}_{o_2t}$.
	The observation implies (Lemma \ref{lemma:rightTriangle}):
	\skipSpace
	\begin{center}
	$\eucl{p^{vis}_{so_1}} + \eucl{p^{vis}_{o_2t}} \leq \sqrt{2} \cdot \eucl{p^{vis}_{st}}.$
	\end{center}
	\skipSpace
	The idea is to construct a path in $G_{BB}$ which starts at $s$, walks to $o_1$, from $o_1$ to $o_2$ and finally reaches $t$.
	The path construction of \Cref{section:shortestBBPaths} obtains a path $p^{BB}_{so_1}$ in $G_{BB}$ between $s$ and $o_1$ of length at most $\sqrt{2} \cdot \eucl{p^{vis}_{so_1}}$.
	The same holds for $t$ and $o_2$.
	Hence, the sum of path lengths in $G_{BB}$ between $s$ and $o_1$ as well as between $o_2$ and $t$ can be upper bounded as follows:
	\skipSpace
	\begin{align*}
	&\phantom{space}&\eucl{p^{BB}_{so_1}} + \eucl{p^{BB}_{o_2t}} &\leq 	\sqrt{2} \cdot \left(\eucl{p^{vis}_{so_1}} + \eucl{p^{vis}_{o_2t}}  \right)\\
	&&& \leq \sqrt{2} \cdot \sqrt{2} \cdot \eucl{p^{vis}_{st}} = 2 \cdot \eucl{p^{vis}_{st}}
	\end{align*}
	\skipSpace
	It remains to prove an upper bound for $\eucl{p^{vis}_{o_1o_2}}$. 
	We use a very rough bound here which is motivated by the triangle inequality.
	\skipSpace
	\begin{center}
	$\eucl{p^{vis}_{o_1o_2}} \leq \eucl{p^{vis}_{so_1}} + \eucl{p^{vis}_{st}} + \eucl{p^{vis}_{o_2t}}$
	\end{center}
\skipSpace
	Note that $p^{BB}_{o_1o_2}$ and $p^{vis}_{o_1o_2}$ coincide due to the structure of modified Bounding Box Visibility Graphs.
	Hence, we conclude $p^{BB}_{o_1o_2}$ has length at most:
	\skipSpace
	\begin{center}
	$\eucl{p^{BB}_{o_1o_2}} \leq \sqrt{2} \cdot \eucl{p^{vis}_{st}} + \eucl{p^{vis}_{st}} = \left(1+\sqrt{2}\right) \cdot \eucl{p^{vis}_{st}}$
	\end{center}
\skipSpace
	Finally, we obtain the following bound on $p^{BB}_{st}$:
	\skipSpace
	\begin{align*}
	&\phantom{space}&\eucl{p^{BB}_{st}} &= \eucl{p^{BB}_{so_1}} + \eucl{p^{BB}_{o_1o_2}} + \eucl{p^{BB}_{o_2t}} \\
	&&&\leq 2 \cdot \eucl{p^{vis}_{st}} + \left(1+\sqrt{2} \right) \cdot  \eucl{p^{vis}_{st}} \\
	&&&\leq 4.42 \cdot d_{\mathrm{UDG}}(s,t).
	\end{align*}
	This bound can be applied for all sub-paths of $p_{vis}$. 
	In cases where $p_{vis}$ does not lead through an area of intersecting bounding boxes, Theorem \ref{theorem:nonIntersectingBB} holds and we obtain an upper bound of $4.42$ for either case.  \\
\end{proof}
\noindent
Note that the given upper bound on the path length is not tight. 
Nevertheless, we have proven that the resulting path length is $c$-competitive to a path length in usual Visibility Graphs.
To conclude this section, we combine the results of this section with previous results and conclude a final upper bound for paths in \twoDel{} that use only nodes of bounding boxes as intermediate points. \\

\begin{corollary} \label{corollary:intersecting}
	Let $m \in \mathbb{N}$.
	Consider a \twoDel{} that contains multiple holes whose bounding boxes could intersect. 
	There exists a path $p^{BB}_{st}$ between $s$ and $t$ in the \twoDel{} that walks along nodes which represent points of a modified Bounding Box Visibility Graph with $\eucl{p^{2Del}_{st}} \leq 12.83 \cdot d_{\mathrm{UDG}}(s,t)$. 
\end{corollary}

\begin{proof}
	Due to Theorem \ref{theorem:modifiedBBVisibility}, there exists a path $p^{BB}_{st}$ between any points $s$ and $t$, provided $s$ and $t$ do not lie inside of any bounding box, in modified Bounding Box Visibility Graphs with length at most $4.42 \cdot d_{\mathrm{UDG}}(s,t)$.
	
	\noindent
	The nodes visited by $p^{BB}_{st}$ are usually not part of the \twoDel{}.
	Instead, we use the points of the \twoDel{} which lie geographically closest to the points of $p^{BB}_{st}$.
	For all sub-paths between $s$ and $t$ that do not lead through an area of intersecting bounding boxes, we have proven that there exists a $\mixedChord$-competitive path in the \twoDel{} (Corollary \ref{corollary:multipleBB}).	
	It remains to prove a bound for sub-paths leading through an area of intersecting bounding boxes.
	Consider a pair of points $p_i$ and $p_j$ lying on the path from $s$ to $t$ such that the path between $p_i$ and $p_j$ leads through an area of intersecting bounding boxes.
	The path between $p_i$ and $p_j$ can be split into three sub-paths, 
	the path from $p_i$ to the outer intersection point $o_1$, a path between $o_1$ and a second outer intersection point $o_2$ and a path from $o_2$ to $p_j$.
	In the proof of Theorem \ref{theorem:modifiedBBVisibility} we analyzed:
	\skipSpace
	\begin{enumerate}
		\item $\eucl{p^{BB}_{p_io_1}} + \eucl{p^{BB}_{o_2p_j}} \leq 2 \cdot d_{\mathrm{UDG}}(p_i,p_j)$
		\item $\eucl{p^{BB}_{o_1,o_2}} \leq (1 + \sqrt{2}) \cdot d_{\mathrm{UDG}}(p_i,p_j)$
	\end{enumerate}
	\skipSpace
	$\eucl{p^{BB}_{p_io_1}}$ and $\eucl{p^{BB}_{o_2p_j}}$ do not lead through an area of intersecting bounding boxes.
	Hence, we can apply the results of Corollary \ref{corollary:multipleBB} and obtain a $3.996 \cdot 2 = 7.992$-competitive path in the \twoDel{}.
	$\eucl{p^{BB}_{o_1,o_2}}$ is the Euclidean shortest path between $o_1$ and $o_2$.
	$o_1$ and $o_2$ are also represented by their geographically closest points in the \twoDel{}. 
	Hence, the path between $o_1$ and $o_2$ is at most $2 \cdot \eucl{p^{BB}_{o_1,o_2}}$.
	For the path in the \twoDel{} we obtain:
	\skipSpace
	\begin{align*}
	&\phantom{space}&\eucl{p^{2Del}_{p_i,p_j}} &\leq (7.992 + 2 \cdot (1 + \sqrt{2})) \cdot d_{\mathrm{UDG}}(p_i,p_j) \\
	&&&\leq 12.83 \cdot d_{\mathrm{UDG}}(p_i,p_j)
	\end{align*}
	\skipSpace
	Thus, each sub-path between $s$ and $t$ is at most $12.83$-competitive and we have proven Corollary \ref{corollary:intersecting}.
	
\end{proof}

\section{Overlay Network}
In this section, we discuss how to compute Bounding Box Visibility Graphs in a distributed manner such that nodes of the ad hoc network are enabled to find $c$-competitive paths.
We make use of a similar approach as in \cite{algosensorsPaper}. Nodes that are part of a boundary (either the outer boundary or a hole) build a hypercube with a pointer jumping technique. 
By this approach it is ensured that all nodes of the same hole are connected in a hypercube.
This hypercube is used for a fast dissemination of bounding box information.
As a bounding box is defined by four extreme coordinates, the nodes of the same hole exchange their coordinates and keep the maximal and the minimal value both of $x$- and $y$-coordinates.
At the end of this procedure, i.e., after $\mathcal{O}(\log n)$ communication rounds, each node is aware of the bounding box coordinates of its hole. 
Afterwards, the node with the smallest $x$-coordinate is responsible for exchanging bounding box information with other bounding boxes.
To do so, an overlay tree according to the protocol in~\cite{DBLP:conf/icalp/GmyrHSS17} is built initially.
The initial setup of the tree requires $\mathcal{O}(\log^2 n)$ communication rounds.
The node with smallest $x$-coordinate then propagates the bounding box information of its hole up and down in the tree.
The representatives of a bounding box are responsible for the computation of $c$-competitive paths in the ad hoc network and store
the modified Bounding Box Visibility Graph locally. 
All other nodes only store whether they are located on the boundary of a bounding box or not.

\section{Routing Algorithm} \label{section:bbr}
In this section, we introduce our routing algorithm \emph{Bounding Box Routing} (\textit{BBR}).
We assume only source-destination pairs which lie outside of bounding boxes.
\subsection{Non-intersecting Bounding Boxes}

In this case, the routing algorithm works as follows:
A source node~$s$ that wants to send a message to a target node $t$ starts sending the packet via the MixedChorArc-algorithm.
In case the packet arrives at an edge of a bounding box, the packet is redirected to the closest representative of the bounding box.
The representative then computes a path from itself to $t$ in its modified Bounding Box Visibility Graph.
This path is used to route the packet to $t$.
Along every edge of the Bounding Box Visibility Graph, MixedChordArc is used to forward the packet.
\emph{BBR} has the following properties:
\skipSpace
\begin{theorem} \label{theorem:bbr}
	BBR finds paths between a source $s$ and a target $t$ outside of bounding boxes with length at
	most $\algorithmConstant \cdot d_{UDG}(s,t)$. 
\end{theorem}

\begin{proof}
	Based on the results of \Cref{section:shortestBBPaths}, online routing the via the path in the Bounding Box Visibility Graph yields a path of length $\onlineRoutingConstant \cdot d_{UDG}(s,t)$.
	BBR, however, used an initial optimistic application of MixedChordArc.
	This phase can result in a detour.
	The detour, however, has length at most $3 \cdot \mixedChord \cdot d_{UDG(s,t)} = 10.68 \cdot d_{UDG}(s,t)$, as the path until a bounding box edge is reached has length at most $\mixedChord \cdot d_{UDG}(s,t)$.
	The same holds for the path to the closed representative of a bounding box.
	Afterwards, it could happen that the path initially returns to $s$ and afterwards leads to $t$.
	In sum, this is a detour of $10.68 \cdot d_{UDG}(s,t)$.
	The complete path is $10.68+7.87 = 18.55$-competitive.
	
\end{proof}
\subsection{Pairwise Intersecting Bounding Boxes}
In this section, we weaken the assumption of non-intersecting bounding boxes and assume that bounding boxes can pairwise intersect.
The routing procedure in this case works basically exactly as for non-intersecting bounding boxes (Section \ref{section:bbr}), the only difference is that the
Bounding Box Visibility Graph is replaced by a modified Bounding Box Visibility Graph (see Section \ref{section:intersectingBB}).
In the modified Bounding Box Visibility Graph we now have pairwise intersecting bounding boxes.
To each pair of pairwise intersecting bounding boxes, there are two outer intersection points contained and connected via an edge.
In this section, we introduce PairwiseIntersectionChord (PIC) that allows us to route $\sqrt{2}$-competitive between two outer intersection points of pairwise intersecting bounding boxes.
Given the competitive constant of PIC, we select the weights of edges in the modified Bounding Box Visibility Graph as follows:
Let $e$ be an edge connecting two outer intersection points $o_1$ and $o_2$. 
In the modified Bounding Box Visibility Graph, we add $\sqrt{2} \cdot \|o_1\, o_2\|$ as weight to $e$.

\begin{theorem}\label{theorem:pairwiseIntersecting}
	In case the modified Bounding Box Visibility Graph contains pairwise intersecting bounding boxes and the edge weight between outer intersection points is at most $\sqrt{2}$ times the length of the shortest path between these points, 
	BBR finds paths between a source $s$ and a target $t$ outside of bounding boxes with length at
	most $28.83 \cdot d_{UDG}(s,t)$. 
\end{theorem}

\begin{proof}
	The proof is analogous to the proof of \Cref{theorem:bbr}. 
	We also have the initial detour of $10.68 \cdot d_{UDG}(s,t)$
	Now, we can apply \Cref{corollary:intersecting}. 
	The path we make use of is $\sqrt{2} \cdot 12.83 \leq 18.15$-competitive. (because of the additional $\sqrt{2}$) obtained by PIC.
	Hence, the entire path has length at most $10.68 + 18.15 = 28.83 \cdot d_{UDG}(s,t)$.
	
\end{proof}

We continue with introducing PIC in detail.
The rough idea behind this algorithm is to route along the convex hulls of holes belonging to outer intersection points.
Convex hulls can be easily calculated in a preprocessing step as done in \cite{algosensorsPaper}.
Certain edges of the convex hulls will give us the direction to route through the pairwise intersection of bounding boxes. 
More precisely, PIC consists of the following $3$ steps.
We call the source outer intersection point by $bb_{int1}$ and the destination point $bb_{int4}$.
Without loss of generality, we assume that the bounding boxes aligned as in Figure \ref{fig:new} and denote the upper left bounding box belonging to a polygon $p_1$ as the \emph{upper} bounding box and the other one belonging to a polygon $p_2$ as the lower bounding box.
We also assume that we route from the left outer intersection point to the right one.
All other cases are simply rotated and are handled analogously.

\begin{enumerate}
	\item $bb_{int1}$ sends a message m along the edge $\overline{bb_{b\ell}(p_1)bb_{br}(p_1)}$ of the upper bounding box until $m$ reaches the intersection point $bb_{int2}$ of $\overline{bb_{b\ell}(p_1)bb_{br}(p_1)}$ and a convex hull.
	\item Then from $bb_{int2}$ $m$ is sent along the convex hull edges until it $m$ arrives at the point $bb_{int3}$, which is the intersection point between a convex hull edge and $\overline{bb_{tr}(p_2),bb_{br}(p_2)}$.
	\item $bb_{int3}$ sends $m$ to the destination point $bb_{int4}$, which denotes the intersection point of $\overline{bb_{tr}(p_2),bb_{br}(p_2)}$ from the upper bounding box and $\overline{bb_{t\ell}(p)bb_{tr}(p)}$ from the lower bounding box.
\end{enumerate}

\begin{figure}[h]
	\centering
	\includegraphics[width=0.75\textwidth]{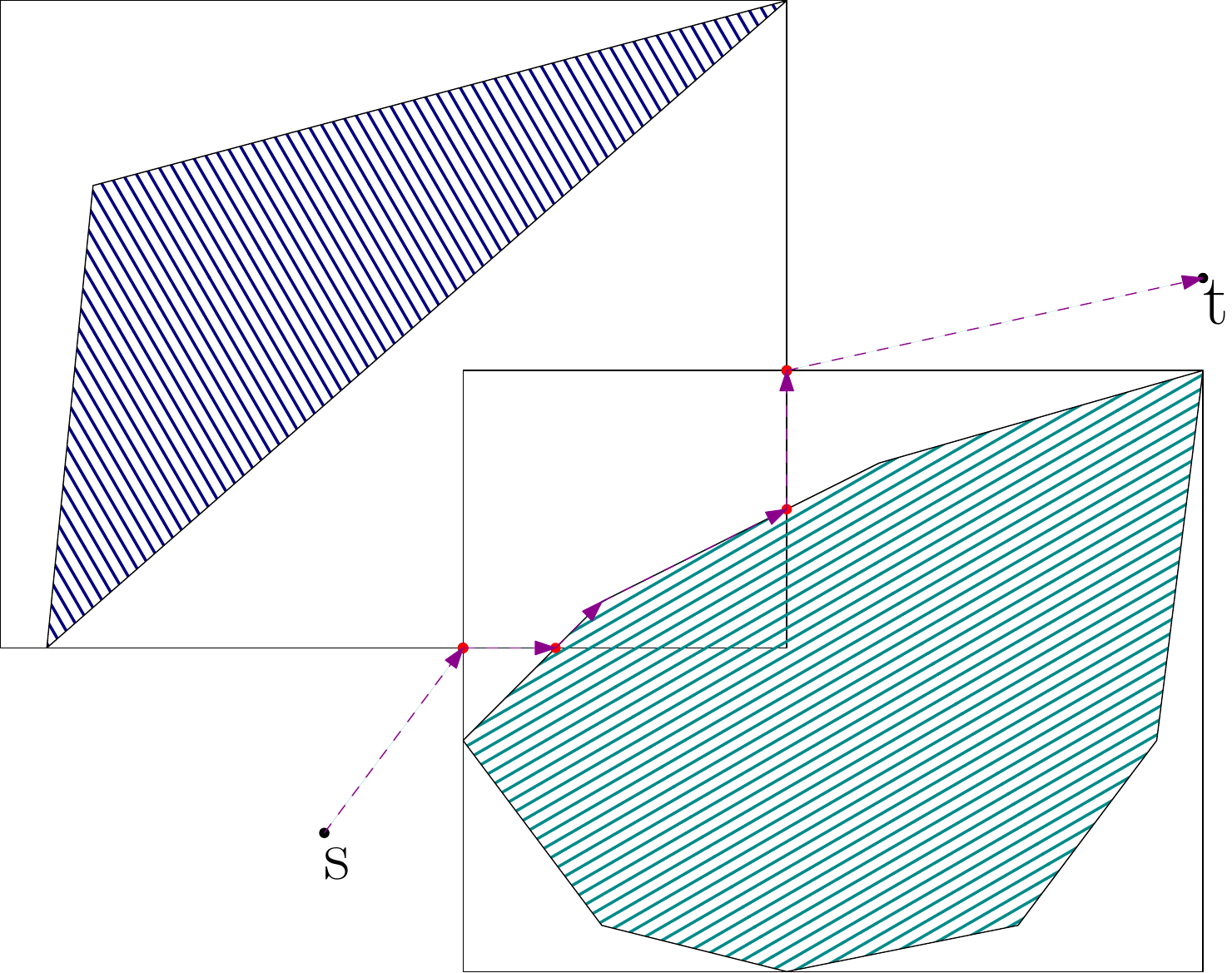}
	\caption{An exemplary path constructed by PIC between two outer intersection points.}
	\label{fig:new} 
	\hfill
	
\end{figure}

To route $m$ as in the $3$ steps above, we use MixedChordArc. 
Therefore we obtain the following theorem.
\begin{theorem} \label{theorem:pairintbb}
	PIC finds paths path between two outer intersection points $o_1$ and $o_2$ of pairwise intersecting bounding boxes with length at most
	$\sqrt{2} \cdot \|o_1\,o_2\|$.
\end{theorem}

\subsection{Multiple Intersecting Bounding Boxes}
The algorithm works exactly as for the case with non-intersecting bounding boxes.
However, we consider modified Bounding Box Visibility Graphs for the computation of paths here.
In case we know a weight between outer intersection points of the modified Bounding Box Visibility Graph which has length at most $c$ times the length of an optimal path between the outer intersection points,
we can state with the same arguments as for \Cref{theorem:bbr} the following theorem:

\begin{theorem} \label{theorem:bbrIn}
	In case the modified Bounding Box Visibility Graph contains an edge weight between outer intersection points which is at most $c$ times the length of the shortest path between these points, 
	BBR finds paths between a source $s$ and a target $t$ outside of bounding boxes with length at
	most $\left(10.68 + c \cdot 12.83\right)\cdot d_{UDG}(s,t)$. 
	
\end{theorem}

\begin{proof}
	The proof is analogous to the proof of \Cref{theorem:pairwiseIntersecting}.
\end{proof}
In our scenario, however, we do not know such an edge weight for outer intersection points.
Since the computation of a shortest or $c$-competitive path between outer intersection points can be very complex (see \Cref{fig:halloween2}), we propose the following strategy here:
We make use of the strategy GOAFR+ between outer intersection points. Let $o_1$ and $o_2$ be two outer intersection points of the same hole.
We assign $\alpha \cdot \|o_1o_2\|^2$ as a weight for the edge $\{o_1,o_2\}$ in the modified Bounding Box Visibility Graph.
This is done, since GOAFR+ is in worst cases quadratic competitive to a shortest possible path.
We show in \Cref{section:simulations} that this approach is much better than using GOAFR+ for the entire path.
We also analyse suitable values for $\alpha$.
All in all, we use the same procedure as for non-intersecting bounding boxes.
The difference is the weight of edges between outer intersection points in the modified Bounding Box Visibility Graph.
Whenever an edge between outer intersection points is part of a path, GOAFR+ instead of MixedChordArc is used.
\section{Simulation Results} \label{section:simulations}
In this section, we compare our routing strategy to existing strategies and prove that we outperform these strategies both for intersecting and non-intersecting bounding boxes.
With outperforming, we mean that the path length found by \emph{BBR} are significantly shorter than those found by strategies that do not consider any global knowledge about holes.
To be able to directly compare our results to earlier results, we choose the same experiment setup as in \cite{DBLP:conf/mobihoc/KuhnWZ03}.
Our simulations are carried out on randomly generated Unit Disk Graph.
We place nodes randomly and uniformly on a square with $20$ units side length.
The number of nodes depends on the network density.
Density means the number of nodes per Unit Disk.
As we are only interested in connected networks, the lowest density value we consider is 4.5. 
Even for the value $4.5$ not all randomly and uniformly generated networks are connected. 
Since all approaches do not work for disconnected networks, we decided to run simulations only on connected networks.
For each density value between $4.5$ and $20$, we generate $2000$ networks and choose a source node $s$ and a target node $t$ (outside of bounding boxes) uniformly at random.
The simulations were carried out on a custom simulation environment.
We compute the following performance value for every chosen algorithm:

\begin{align*}
&&perf_A(N,s,t) = \frac{\|p_A(N,s,t)\|}{\|p_{opt}(N,s,t)\|}
\end{align*}
\skipSpace
$|p_A(N,s,t)\|$ represents the length of the path found by algorithm $A$ and $\|p_{opt}(N,s,t)\|$ stands for the length of an optimal path contained in the network $N$.
Before switching the results of the simulations, we introduce compared algorithms.

\subsection{Compared Algorithms}
In the following, we shortly introduce the algorithms 

\paragraph*{OAFR} Let $\mathcal{E}(c)$ be the ellipse with foci $s$ and $t$ and the size of its major axis is $c$. 
Initially, set $\mathcal{E}$ to $2 \cdot \|st\|$. Explore the boundary of the face F that is intersected by $\overline{st}$ in one direction. 
Upon reaching the boundary of $\mathcal{E}$, switch the direction. Once $\mathcal{E}$ is hit a second time, proceed to the node closest to $t$.
In case $t$ is not yet reached, double the length of $\mathcal{E}'s$ major axis and start again.
\paragraph*{GOAFR}
Initially, set $\mathcal{E}$ to $2 \cdot \|st\|$.  
Apply greedy routing until either arriving at $t$ oder a local minimum $m$.
Execute OAFR on the first face only.
The size of $\mathcal{E}$ is doubled as long as necessary.
Terminate if OAFR reaches $t$.
Otherwise, approach to the node closest to $t$ found by OAFR and start again.
\paragraph*{GOAFR\textsubscript{FC}}
Works basically like GOAFR but falls back to greedy routing as soon as a closer node to $t$ is found in the OAFR-phase \cite{DBLP:conf/mobihoc/KuhnWZ03}.
\paragraph*{GOAFR+}
GOAFR+ works similar to GOAFR but the size of the major axis of $\mathcal{E}$ is not doubled.
Let $m$ be the node where GOAFR switches from greedy routing to OAFR.
Now, GOAFR continues and counts in a variable $p$ the number of nodes closer to $t$ than $u_i$ and $q$ the number of other nodes.
Let $r_\mathcal{E}$ be the major axis of $\mathcal{E}$. 
The new $r_\mathcal{E}$ is $p \cdot r_\mathcal{E}$.

\paragraph*{GPSR} 
GPSR stands for Greedy Perimeter Stateless Routing. It works exactly as GOAFR\textsubscript{FC} without bounding ellipse.
\paragraph*{BBR} See \Cref{section:bbr}
\subsection{Non-intersecting Bounding Boxes}
\Cref{fig:simulationNonIntersecting} shows the simulation results in case of non-intersecting bounding boxes. 
For each density value, the mean performance value of each algorithm, defined as

\begin{center}
$\overline{perf_A}:= \frac{1}{k} \sum_{i=1}^{k} perf_A(N_i,s_i,t_i).$
\end{center} is plotted.
In our simulations, $k = 2000$.
\begin{figure}[h]
	\centering
	\includegraphics[width=\textwidth]{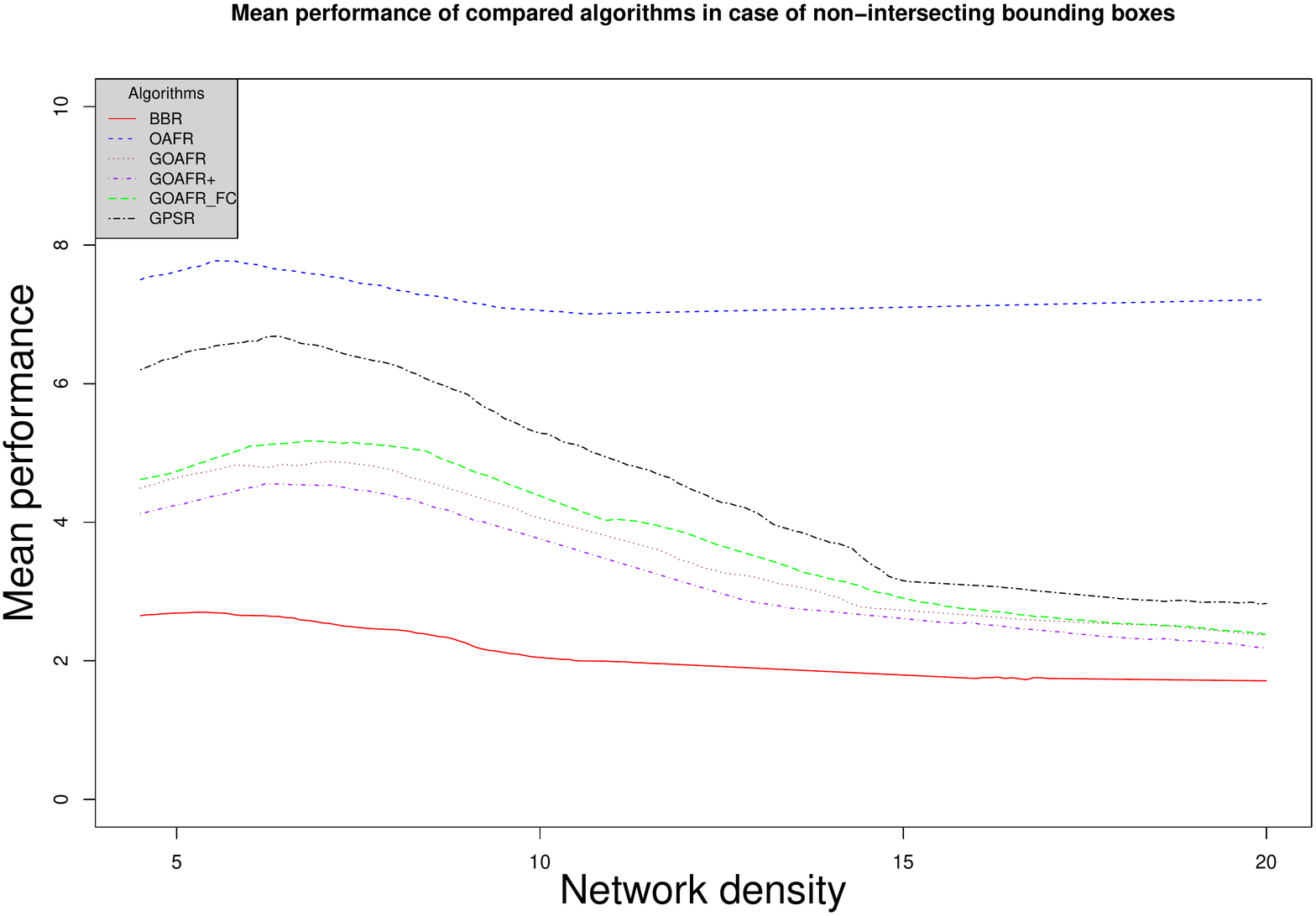}
	\caption{Visualization of the performance of the algorithms in case of non-intersecting bounding boxes. The y-axis represents the performance value of the algorithms.
		The x-axis represents the network density.
		The lowest red line represents the performance of BBR.}
	\label{fig:simulationNonIntersecting}
\end{figure}
\skipSpace
We can see that for each density value, BBR performs much better than any of the other algorithms.
Even in cases where many holes exists (density values around 4.5), avoiding holes around their bounding boxes seems to perform much butter than any strategy that only considers local knowledge.
Also for scenarios with few holes (high density values), BBR still performs slightly better than GOAFR+.
Interestingly, all algorithms have a small peak at density values around 8.
\subsection{Intersecting Bounding Boxes}
In order to really get an intuition about the performance in case of intersecting bounding boxes, we chose only source destination pairs whose shortest path leaded through an area of intersecting bounding boxes.
\Cref{fig:simulationIntersecting} shows a plot of the mean performance of all algorithms depending on the density value of the network.

\begin{figure}[h]
	\centering
	\includegraphics[width=\textwidth]{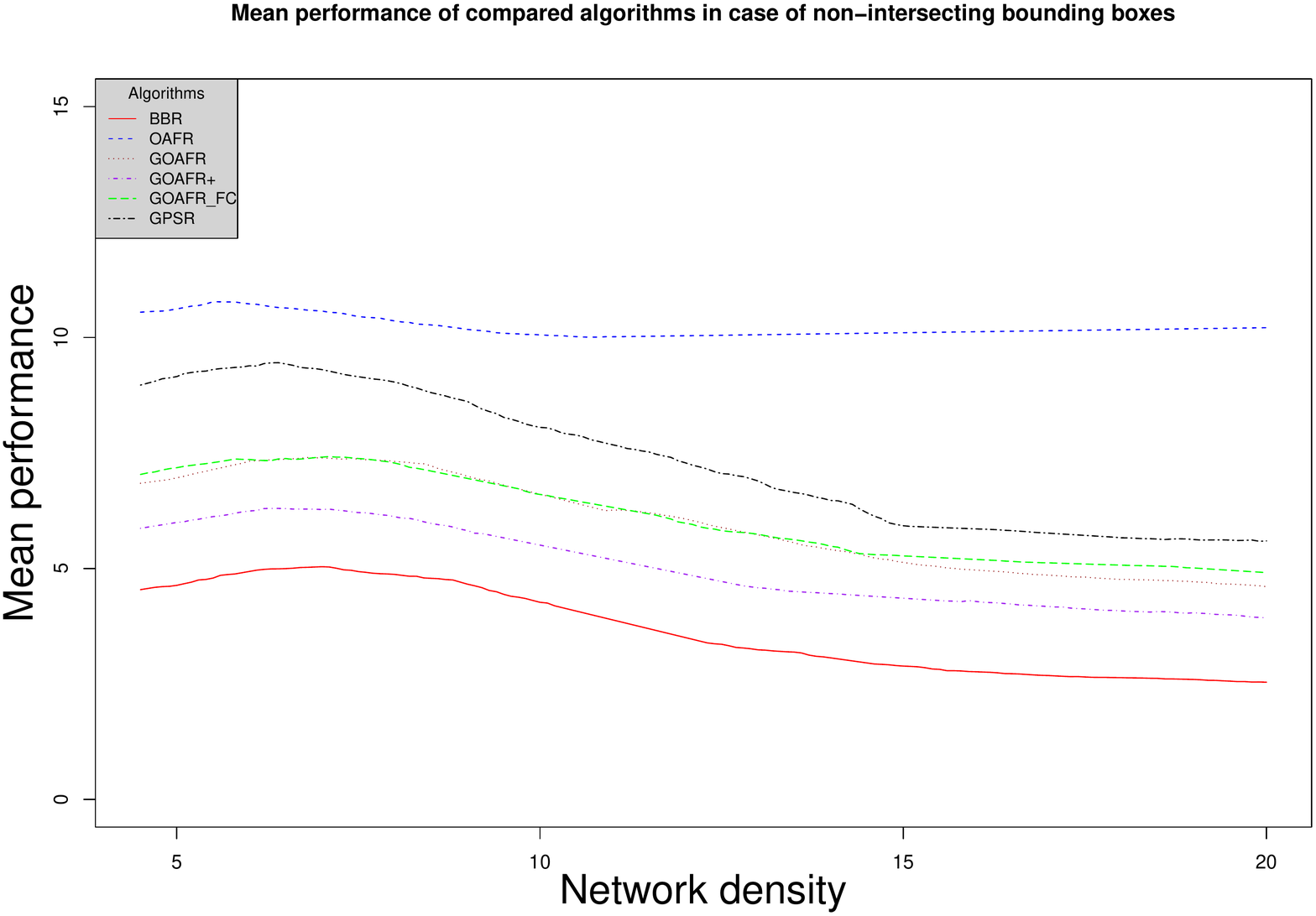}
	\caption{Visualization of the performance of the algorithms in case of intersecting bounding boxes. The y-axis represents the performance value of the algorithms.
		The x-axis represents the network density.
		The lowest red line represents the performance of BBR.}
	\label{fig:simulationIntersecting}
\end{figure}

In \Cref{fig:simulationIntersecting}, we can see that BBR still performs significantly better than all other approaches.
However, the line of BBR is much closer to the line of GOAFR+ compared to non-intersecting bounding boxes.
This probably results from the fact that BBR uses GOAFR+ between outer intersection points.

\section{Future Work}
In this paper, we considered both intersecting and non-intersecting bounding boxes as hole abstractions.
For non-intersecting bounding boxes, we designed a routing strategy to route in the underlying ad hoc network which achieves provably better results than any previous online routing strategy for geometric ad hoc networks.
Our strategy routing finds $18.55$-competitive paths between source and destination pairs outside of bounding boxes.
In case of intersecting bounding boxes, we considered the case of pairwise intersecting bounding boxes and presented and proved a routing strategy that finds $28.83$-competitive paths between any pair of sources and destinations outside of bounding boxes. In addition, in case of multiple intersecting bounding boxes, we proposed a $(10.68 + c \cdot 12.83)$-competitive routing strategy, provided $c$-competitive paths between representatives of outer intersection points of bounding boxes can be found.
Since this is a very challenging problem, our goal is to investigate how many knowledge is necessary to find $c$-competitive routing paths in areas of intersecting hole abstractions.
Further interesting directions include energy efficient routing and load balancing to avoid congested bounding box nodes.




\bibliography{boundingBib}

\begin{thebibliography}{10}

\bibitem{Ahmed:2005:HPW:1072989.1072992}
Nadeem Ahmed, Salil~S. Kanhere, and Sanjay Jha.
\newblock The holes problem in wireless sensor networks: A survey.
\newblock {\em SIGMOBILE Mob. Comput. Commun. Rev.}, 9(2):4--18, April 2005.
\newblock URL: \url{http://doi.acm.org/10.1145/1072989.1072992}, \href
  {http://dx.doi.org/10.1145/1072989.1072992}
  {\path{doi:10.1145/1072989.1072992}}.

\bibitem{computationalGeometryDeBerg}
Mark~de Berg, Otfried Cheong, Marc~van Kreveld, and Mark Overmars.
\newblock {\em Computational {G}eometry: {A}lgorithms and {A}pplications}.
\newblock Springer-Verlag TELOS, Santa Clara, CA, USA, 3rd edition, 2008.

\bibitem{DBLP:conf/esa/BonichonBCDHS18}
Nicolas Bonichon, Prosenjit Bose, Jean{-}Lou~De Carufel, Vincent Despr{\'{e}},
  Darryl Hill, and Michiel H.~M. Smid.
\newblock Improved routing on the delaunay triangulation.
\newblock In {\em 26th Annual European Symposium on Algorithms, {ESA} 2018,
  August 20-22, 2018, Helsinki, Finland}, pages 22:1--22:13, 2018.
\newblock URL: \url{https://doi.org/10.4230/LIPIcs.ESA.2018.22}, \href
  {http://dx.doi.org/10.4230/LIPIcs.ESA.2018.22}
  {\path{doi:10.4230/LIPIcs.ESA.2018.22}}.

\bibitem{competitiveRoutingInDelaunayImproved}
Prosenjit Bose, Jean-Lou De~Carufel, Stephane Durocher, and Perouz Taslakian.
\newblock {\em Competitive Online Routing on Delaunay Triangulations}, pages
  98--109.
\newblock Springer International Publishing, 2014.
\newblock \href {http://dx.doi.org/10.1007/978-3-319-08404-6_9}
  {\path{doi:10.1007/978-3-319-08404-6_9}}.

\bibitem{face12}
Prosenjit Bose, Pat Morin, Ivan Stojmenovic, and Jorge Urrutia.
\newblock Routing with guaranteed delivery in ad hoc wireless networks.
\newblock {\em Wireless Networks}, 7(6):609--616, 2001.
\newblock URL: \url{https://doi.org/10.1023/A:1012319418150}, \href
  {http://dx.doi.org/10.1023/A:1012319418150}
  {\path{doi:10.1023/A:1012319418150}}.

\bibitem{hybridWiredWireless}
G.~Cena, A.~Valenzano, and S.~Vitturi.
\newblock Hybrid wired/wireless {N}etworks for real-time {C}ommunications.
\newblock {\em IEEE Industrial Electronics Magazine}, 2(1):8--20, March 2008.

\bibitem{randomForwarding}
Douglas~SJ De~Couto and Robert Morris.
\newblock Location proxies and intermediate node forwarding for practical
  geographic forwarding, 2001.

\bibitem{DBLP:conf/infocom/FangGG04}
Qing Fang, Jie Gao, and Leonidas~J. Guibas.
\newblock Locating and bypassing routing holes in sensor networks.
\newblock In {\em Proceedings {IEEE} {INFOCOM} 2004, The 23rd Annual Joint
  Conference of the {IEEE} Computer and Communications Societies, Hong Kong,
  China, March 7-11, 2004}, pages 2458--2468, 2004.
\newblock URL: \url{https://doi.org/10.1109/INFCOM.2004.1354667}, \href
  {http://dx.doi.org/10.1109/INFCOM.2004.1354667}
  {\path{doi:10.1109/INFCOM.2004.1354667}}.

\bibitem{DBLP:conf/icalp/GmyrHSS17}
Robert Gmyr, Kristian Hinnenthal, Christian Scheideler, and Christian Sohler.
\newblock Distributed monitoring of network properties: The power of hybrid
  networks.
\newblock In {\em 44th International Colloquium on Automata, Languages, and
  Programming, {ICALP} 2017, July 10-14, 2017, Warsaw, Poland}, pages
  137:1--137:15, 2017.
\newblock URL: \url{https://doi.org/10.4230/LIPIcs.ICALP.2017.137}, \href
  {http://dx.doi.org/10.4230/LIPIcs.ICALP.2017.137}
  {\path{doi:10.4230/LIPIcs.ICALP.2017.137}}.

\bibitem{algosensorsPaper}
Daniel Jung, Christina Kolb, Christian Scheideler, and Jannik Sundermeier.
\newblock Competitive routing in hybrid communication networks.
\newblock In {\em 14th International Symposium on Algorithms and Experiments
  for Wireless Networks {Algosensors} 2018, August 23-24, 2018, Helsinki,
  Finland}, 2018 in press.

\bibitem{gpsr}
Brad Karp and H.~T. Kung.
\newblock {GPSR:} greedy perimeter stateless routing for wireless networks.
\newblock In {\em {MOBICOM} 2000, Proceedings of the sixth annual international
  conference on Mobile computing and networking, Boston, MA, USA, August 6-11,
  2000.}, pages 243--254, 2000.
\newblock URL: \url{http://doi.acm.org/10.1145/345910.345953}, \href
  {http://dx.doi.org/10.1145/345910.345953} {\path{doi:10.1145/345910.345953}}.

\bibitem{compass}
Evangelos Kranakis, Harvinder Singh, and Jorge Urrutia.
\newblock Compass routing on geometric networks.
\newblock In {\em Proceedings of the 11th Canadian Conference on Computational
  Geometry, UBC, Vancouver, British Columbia, Canada, August 15-18, 1999},
  1999.
\newblock URL: \url{http://www.cccg.ca/proceedings/1999/c46.pdf}.

\bibitem{DBLP:conf/podc/KuhnWZZ03}
Fabian Kuhn, Roger Wattenhofer, Yan Zhang, and Aaron Zollinger.
\newblock Geometric ad-hoc routing: of theory and practice.
\newblock In {\em Proceedings of the Twenty-Second {ACM} Symposium on
  Principles of Distributed Computing, {PODC} 2003, Boston, Massachusetts, USA,
  July 13-16, 2003}, pages 63--72, 2003.
\newblock URL: \url{http://doi.acm.org/10.1145/872035.872044}, \href
  {http://dx.doi.org/10.1145/872035.872044} {\path{doi:10.1145/872035.872044}}.

\bibitem{DBLP:conf/dialm/KuhnWZ02}
Fabian Kuhn, Roger Wattenhofer, and Aaron Zollinger.
\newblock Asymptotically optimal geometric mobile ad-hoc routing.
\newblock In {\em Proceedings of the 6th International Workshop on Discrete
  Algorithms and Methods for Mobile Computing and Communications {(DIAL-M}
  2002), Atlanta, Georgia, USA, September 28-28, 2002}, pages 24--33, 2002.
\newblock URL: \url{http://doi.acm.org/10.1145/570810.570814}, \href
  {http://dx.doi.org/10.1145/570810.570814} {\path{doi:10.1145/570810.570814}}.

\bibitem{DBLP:conf/mobihoc/KuhnWZ03}
Fabian Kuhn, Roger Wattenhofer, and Aaron Zollinger.
\newblock Worst-case optimal and average-case efficient geometric ad-hoc
  routing.
\newblock In {\em Proceedings of the 4th {ACM} Interational Symposium on Mobile
  Ad Hoc Networking and Computing, MobiHoc 2003, Annapolis, Maryland, USA, June
  1-3, 2003}, pages 267--278, 2003.
\newblock URL: \url{http://doi.acm.org/10.1145/778415.778447}, \href
  {http://dx.doi.org/10.1145/778415.778447} {\path{doi:10.1145/778415.778447}}.

\bibitem{localDelaunay}
Xiang-Yang Li, G.~Calinescu, and Peng-Jun Wan.
\newblock Distributed {C}onstruction of a {P}lanar {S}panner and {R}outing for
  {A}d {H}oc {W}ireless {N}etworks.
\newblock In {\em Proceedings of the 21st Annual Joint Conference of the IEEE
  Computer and Communications Societies}, volume~3, pages 1268--1277 vol.3, New
  York, NY, USA, 2002. {IEEE} Press.

\bibitem{710380}
Y.~S.~N. Murty.
\newblock Hybrid {C}ommunication {N}etworks for {P}ower {U}tilities.
\newblock In {\em Power Quality '98}, pages 239--242, New York, NY, USA, Jun
  1998. {IEEE} press.

\bibitem{ruhrup2006online}
Stefan R{\"u}hrup and Christian Schindelhauer.
\newblock Online multi-path routing in a maze.
\newblock In Tetsuo Asano, editor, {\em Algorithms and Computation}, pages
  650--659, Berlin, Heidelberg, 2006. Springer Berlin Heidelberg.

\bibitem{xiaDelaunaySpanner}
Ge~Xia.
\newblock The stretch factor of the delaunay triangulation is less than 1.998.
\newblock {\em {SIAM} J. Comput.}, 42(4):1620--1659, 2013.

\bibitem{yu2001geographical}
Yan Yu, Ramesh Govindan, and Deborah Estrin.
\newblock Geographical and energy aware routing: A recursive data dissemination
  protocol for wireless sensor networks.
\newblock 2001.

\end{thebibliography}

\end{document}